\documentclass[11pt]{article}
\usepackage[margin=1in]{geometry}
\usepackage{amsmath,amssymb,amsthm,mathtools}
\usepackage{microtype}
\usepackage{enumitem}
\usepackage{xcolor}
\usepackage{booktabs}
\usepackage{tabularx}
\usepackage{array}
\usepackage{float}
\usepackage{tcolorbox}
\tcbuselibrary{breakable}
\usepackage{cite}
\usepackage{hyperref}

\definecolor{darkblue}{RGB}{20,60,120}
\definecolor{darkgreen}{RGB}{35,110,65}

\hypersetup{
  colorlinks=true,
  linkcolor=darkblue,
  citecolor=darkgreen,
  urlcolor=darkblue
}
\newtheorem{theorem}{Theorem}[section]
\newtheorem{lemma}[theorem]{Lemma}
\newtheorem{proposition}[theorem]{Proposition}

\newtheorem*{corollary*}{Corollary}

\theoremstyle{definition}
\newtheorem{definition}[theorem]{Definition}
\newtheorem{remark}[theorem]{Remark}

\newtcolorbox[auto counter]{procedure}[2]{
  breakable,sharp corners,boxrule=.5pt,boxsep=0pt,
  left=8pt,right=8pt,top=6pt,bottom=6pt,toptitle=4pt,bottomtitle=4pt,
  colback=white,colbacktitle=black!7,colframe=black!65,coltitle=black,
  fonttitle=\bfseries,title={Algorithm~\thetcbcounter\quad #1},
  title after break={Algorithm~\thetcbcounter\quad #1 \normalfont(continued)},
  label={#2},lines before break=5
}
\newcommand{\procedureio}[2]{\noindent\textbf{Input}\enspace #1\par
  \noindent\textbf{Output}\enspace #2\par}

\newcommand{\rv}[1]{\boldsymbol{#1}}
\newcommand{\E}{\mathbb{E}}
\newcommand{\Prb}{\mathbb{P}}
\newcommand{\vol}{\operatorname{vol}}

\newcommand{\poly}{\operatorname{poly}}

\newcommand{\1}{\mathbf{1}}

\newcommand{\tO}{\widetilde O}

\title{Tolerant Testing for Bipartiteness and Unique Games}
\author{Yuichi Yoshida\thanks{Supported by JSPS KAKENHI Grant Number 22H05001 and 25K24465.}\\
National Institute of Informatics\\
\texttt{yyoshida@nii.ac.jp}}
\RequirePackage[normalem]{ulem} 
\RequirePackage{color}\definecolor{RED}{rgb}{1,0,0}\definecolor{BLUE}{rgb}{0,0,1} 
\makeatletter 
\let\sout@orig\sout 
\renewcommand{\sout}[1]{\ifmmode\text{\sout@orig{\ensuremath{#1}}}\else\sout@orig{#1}\fi} 
\makeatother 
\RequirePackage{listings} 
\RequirePackage{color} 
\lstdefinelanguage{DIFcode}{ 
  moredelim=[il][\color{red}\sout]{\%DIF\ <\ }, 
  moredelim=[il][\color{blue}\uwave]{\%DIF\ >\ } 
} 
\lstdefinestyle{DIFverbatimstyle}{ 
	language=DIFcode, 
	basicstyle=\ttfamily, 
	columns=fullflexible, 
	keepspaces=true 
} 
\lstnewenvironment{DIFverbatim}{\lstset{style=DIFverbatimstyle}}{} 
\lstnewenvironment{DIFverbatim*}{\lstset{style=DIFverbatimstyle,showspaces=true}}{} 
\begin{document}
\maketitle

\begin{abstract}
We give a tolerant bipartiteness tester in the adjacency-list model on
arbitrary simple undirected graphs.  Let \(n\) be the number of vertices
and \(m>0\) the known number of edges.  For \(0<\varepsilon<\rho<1/10\),
write \(\lambda_\rho:=\rho/(1+\log(1/\rho))\).  If
\(\varepsilon\log n\le c\lambda_\rho^2\) for a sufficiently small absolute
constant \(c>0\), the tester distinguishes graphs that can be made bipartite
by deleting at most an \(\varepsilon\) fraction of their edges from graphs
requiring deletion of at least a \(\rho\) fraction, using
\[
 \widetilde O\!\left(
 \frac{\sqrt{n/\rho}}{\lambda_\rho^{3/2}}
 +\frac{1}{\lambda_\rho^2}
 +\frac{n}{\sqrt m\,\lambda_\rho}\right)
\]
queries, where the suppressed factors are polylogarithmic in \(n,1/\rho\).
For fixed \(\rho\), the query bound is
\(\widetilde O(\sqrt n+n/\sqrt m)\).  We prove an
\(\Omega(\sqrt n+n/\sqrt m)\) query lower bound.

The tester extends to Unique Games over any alphabet size \(Q\ge2\),
under the same tolerance condition with \(c\) independent of \(Q\),
with a polylogarithmic overhead in \(Q\) in the query complexity.
Each constrained neighbor query returns the entire edge permutation.
\end{abstract}

\thispagestyle{empty}
\newpage

\thispagestyle{empty}
\tableofcontents

\newpage
\setcounter{page}{1}

\section{Introduction}
\label{sec:overview-main}

In graph property testing, a sublinear-query algorithm uses local
information to distinguish graphs that have a property from graphs far
from having it.
See the standard texts of Goldreich and Bhattacharyya--Yoshida for general background
on property testing~\cite{GoldreichBook,BhattacharyyaYoshidaBook}.

Bipartiteness is a central example.  Goldreich--Ron gave a random-walk-based
tester in the bounded-degree model~\cite{GR}, and
Kaufman--Krivelevich--Ron extended testing to the general graph
model~\cite{KKR}.

Can such testers also accept graphs that become bipartite after deleting
a small fraction of their edges?  This is the question of tolerant
testing~\cite{PRR}.  For a graph \(G=(V,E)\) with \(|E|>0\), write
\[
 \tau(G):=\min_{\chi:V\to\{\pm1\}}
 \frac{|\{uv\in E:\chi(u)=\chi(v)\}|}{|E|},
\]
the minimum uncut fraction, or equivalently the least fraction of edges
whose deletion makes the graph bipartite.  If \(|E|=0\), set \(\tau(G):=0\).
For thresholds \(0\le\varepsilon<\rho\le1\), an
\emph{\((\varepsilon,\rho)\)-tester for bipartiteness} is a randomized
algorithm that accepts every graph with \(\tau(G)\le\varepsilon\) with
probability at least \(2/3\) and rejects every graph with
\(\tau(G)\ge\rho\) with probability at least \(2/3\).
The maximum cut fraction is \(1-\tau(G)\), so this is also a gap version
of MaxCut.

Prior progress on this problem has used structural assumptions.
Peng--Yoshida give tolerant testers on expanders~\cite{PengYoshida}, and
Jha--Kumar extend this line of work to bounded-degree clusterable
graphs~\cite{JhaKumar}.  Extending random-walk-based testing to graphs
without rapid mixing presents additional difficulties in the tolerant
setting: Jha--Kumar explain why the expander-based approach does not
directly extend even to clusterable graphs~\cite[Section~3]{JhaKumar}.
Fei--Rubinfeld
explicitly ask whether sublinear-query tolerant testing is possible on
general graphs with constant tolerance~\cite[Section~1.3, Problem~1.9]{FeiRubinfeld2026}.
\footnote{Tolerant testing in the dense graph model has also been
studied~\cite{AFKK,GhoshMRS}; that model uses adjacency-matrix queries and
normalizes distance by \(n^2\).}

We give a tolerant bipartiteness tester for arbitrary simple undirected
graphs in the adjacency-list model, without expansion or clusterability
assumptions.  For every fixed rejection threshold
\(0<\rho<1/10\), we obtain an \((\varepsilon,\rho)\)-tester whenever
\(\varepsilon\log n\) is sufficiently small as a function of \(\rho\), using
\(\tO(\sqrt n+n/\sqrt m)\) queries, where \(n\) and \(m\) are the
numbers of vertices and edges.
Our method also extends to Unique Games over arbitrary alphabets, as
we describe after the bipartiteness result.

\paragraph{Tolerant bipartiteness.}
The tester is given the number of edges \(m\) and uses uniform-vertex,
degree, and neighbor queries (Definition~\ref{def:adjacency-list-model}).  The bounds below
are for \(m>0\); edgeless inputs are accepted immediately.
Write \(\lambda_\rho:=\rho/(1+\log(1/\rho))\).

\begin{theorem}[Tolerant testing for bipartiteness]
\label{tol:thm:main}
There is an absolute constant \(c>0\) such that the following holds.
Let \(G\) be a simple undirected graph with \(n\) vertices and \(m\) edges.
If \(0<\varepsilon<\rho<1/10\) and
\(\varepsilon\log n\le c\lambda_\rho^2\), there is an
\((\varepsilon,\rho)\)-tester for bipartiteness in the adjacency-list model
with query complexity
\footnote{For our bounds, \(\tO(\cdot)\) hides only polylogarithmic
factors in \(n,1/\rho\).  Dependence on
the alphabet size \(Q\) in the Unique Games extension is explicit.
Dependencies in prior results are specified separately.}
\[
 \tO\!\left(\frac{\sqrt{n/\rho}}{\lambda_\rho^{3/2}}
             +\frac{1}{\lambda_\rho^2}
             +\frac{n}{\sqrt m\,\lambda_\rho}\right).
\]
\end{theorem}

For fixed \(\rho\) and \(m=\Omega(n)\), Theorem~\ref{tol:thm:main} uses
\(\tO(\sqrt n)\) queries at tolerance \(\varepsilon=c\lambda_\rho^2/\log n\).
Could the same query bound hold for a positive constant tolerance?
For every fixed \(0<\varepsilon<\rho<1/10\), Yoshida's lower bound for
2-XOR~\cite[Theorem~1.2]{YoshidaCSP} implies that an
\((\varepsilon,\rho)\)-tester for bipartiteness requires
\(\Omega(n^{1/2+\delta})\) queries for some \(\delta>0\), even on
bounded-degree graphs with \(m=\Theta(n)\)
(see Section~\ref{sec:comparison-prior-work} for the reduction).
Thus the tolerance must decrease with \(n\) to retain \(\tO(\sqrt n)\)
queries.  The lower bound leaves open whether a tolerance of
\(\omega(1/\log n)\) is possible with this query complexity.

Within the tolerance range of Theorem~\ref{tol:thm:main}, the following
lower bound matches its dependence on \(n\) and \(m\) up to logarithmic
factors for fixed \(\rho\).

\begin{proposition}[Query lower bound]
\label{prop:query-lower-bounds}
For every sufficiently large even \(n\), every integer
\(16\le m\le n^2/4\), and every \(0\le\varepsilon<\rho<1/10\), an
\((\varepsilon,\rho)\)-tester for bipartiteness on simple \(n\)-vertex
graphs with exactly \(m\) edges in the adjacency-list model has query
complexity
\[
 \Omega\!\left(\sqrt n+\frac{n}{\sqrt m}\right).
\]
\end{proposition}

\paragraph{Extension to Unique Games.}
Bipartiteness requires the two endpoint labels of each edge to be opposite.
Our method extends to constraints in which either endpoint label uniquely
determines the other, over an arbitrary alphabet.  This is the Unique Games
problem introduced by Khot~\cite{KhotUGC}.
An instance is
\(\mathcal U=(G,\{\pi_{uv}\}_{uv\in E})\), where \(G=(V,E)\) is the
constraint graph with \(n\) vertices and \(m\) edges, and each edge \(uv\)
carries a permutation \(\pi_{uv}\) of the
alphabet \([Q]:=\{1,\ldots,Q\}\) for an integer \(Q\ge2\), with the
reverse direction using
\(\pi_{vu}=\pi_{uv}^{-1}\).  A labeling \(x:V\to[Q]\) satisfies \(uv\) if
\(x(v)=\pi_{uv}(x(u))\).  We allow arbitrary simple undirected constraint
graphs, with one constraint on each edge and possibly isolated vertices.
Extending the notation for bipartiteness, let \(\tau(\mathcal U)\) denote
the minimum fraction of constraints
violated by any labeling, or equivalently the minimum fraction whose
deletion makes the instance satisfiable.  Set
\(\tau(\mathcal U):=0\) when \(m=0\).
An \emph{\((\varepsilon,\rho)\)-tester for Unique Games} has the same acceptance
and rejection guarantees, with \(\tau(\mathcal U)\) in place of \(\tau(G)\).
We use the same access model as above, with a constrained neighbor query
returning the entire edge permutation and counting as one query even when
\(Q\) grows with \(n\).

\begin{theorem}[Tolerant testing for Unique Games]
\label{thm:ug-tolerant}
There is an absolute constant \(c>0\) such that the following holds for
every alphabet size \(Q\ge2\).  Let \(0<\varepsilon<\rho<1\), and let
\(\mathcal U\) be a Unique Games instance over alphabet \([Q]\) on a simple
undirected constraint graph with \(n\) vertices and \(m\) edges.  If
\[
    \varepsilon\log n\le c\lambda_\rho^2,
\]
then there is an \((\varepsilon,\rho)\)-tester for Unique Games in the
adjacency-list model with query complexity
\[
    \tO\!\left(
        \frac{\sqrt{n/\rho}\,(\log Q)^{3/2}}{\lambda_\rho^{3/2}}
        +\frac{\log Q}{\lambda_\rho^2}
        +\frac{n}{\sqrt m\,\lambda_\rho}
    \right).
\]
\end{theorem}

For constant \(Q\), this has the same form as the bipartiteness bound.
For fixed \(\rho\) and \(Q=n^{O(1)}\), the extension retains the
\(\tO(\sqrt n+n/\sqrt m)\) query bound and permits
\(\varepsilon\le c\lambda_\rho^2/\log n\).
The guarantee concerns graph queries; local computation and storage
depend polynomially on \(Q\).

\subsection{Comparison with prior work}
\label{sec:comparison-prior-work}

Table~\ref{tab:related-comparison} summarizes prior testing bounds and our
results, including their testing promises, access models, and structural
assumptions.

\begin{table}[t!]
\caption{Representative quantitative comparison points, with
\(\lambda_\rho=\rho/(1+\log(1/\rho))\); \(c_Q,C_Q>0\) denote constants
depending only on \(Q\).  The conditions for the row marked
\(\dagger\) are specified in the text.  The Goldreich--Ron row uses
incidence-list distance; all other rows use the minimum fraction of violated
edges or constraints.}
\label{tab:related-comparison}
\centering
\small
\setlength{\tabcolsep}{2.2pt}
\renewcommand{\arraystretch}{1.2}
\begin{tabularx}{\textwidth}{@{}
>{\raggedright\arraybackslash}p{0.15\textwidth}
>{\raggedright\arraybackslash}p{0.16\textwidth}
>{\raggedright\arraybackslash}p{0.24\textwidth}
>{\raggedright\arraybackslash}X
>{\raggedright\arraybackslash}p{0.095\textwidth}@{}}
\toprule
Access model & Assumption & Promise / distance & Query complexity & Reference \\
\midrule
\multicolumn{5}{c}{\textsc{Non-tolerant bipartiteness}} \\
bounded-degree
& none
& \(0\) vs.\ \(\rho\)
& \(\sqrt n\,\poly(\log n/\rho)\)
& \cite{GR} \\
general graph
& \(m=\Omega(n)\)
& \(0\) vs.\ constant \(\rho\)
& \(\tO(\min\{\sqrt n,n^2/m\})\)
& \cite{KKR} \\
adjacency-list
& no isolated vertices
& \(0\) vs.\ constant \(\rho\)
& \(O_\rho(\sqrt n\log n)\)
& \cite{FeiRubinfeld2026} \\
\addlinespace
\multicolumn{5}{c}{\textsc{Tolerant bipartiteness}} \\
adjacency-list
& conductance \(\phi\), \(\rho\gtrsim\varepsilon/\phi^2\)
& \(\varepsilon\) vs.\ \(\rho\)
& \(\tO((\phi^2\rho)^{-1}
  m^{1/2+O(\varepsilon/(\phi^2\rho))})\)
& \cite{PengYoshida} \\
adjacency-list
& bounded degree; clusterable\({}^\dagger\)
& \(\varepsilon\) vs.\ \(1/2-\varepsilon\)
& \(\tO(n^{0.501+O(\varepsilon/\phi^2)})\)
& \cite{JhaKumar} \\
adjacency-list
& \(\varepsilon\log n\lesssim \lambda_\rho^2\)
& \(\varepsilon\) vs.\ \(\rho\)
& \(\tO(\sqrt n\,\rho^{-2}+\allowbreak n/(\sqrt m\rho))\)
& Thm.~\ref{tol:thm:main} \\
\midrule
\multicolumn{5}{c}{\textsc{Tolerant Unique Games}} \\
bounded-degree
& conductance \(\phi\),
\(\varepsilon\lesssim \phi^{C_Q}\),
\(\rho\gtrsim \varepsilon^{c_Q}/\phi^4\)
& \(\varepsilon\) vs.\ \(\rho\)
& \(\tO(2^{O(\phi^{\scriptstyle c_Q}/\sqrt\varepsilon)}
   n^{1/2+O(\varepsilon^{\scriptstyle c_Q}/\phi^2)})\)
& \cite{PengYoshida} \\
adjacency-list
& \(\varepsilon\log n\lesssim \lambda_\rho^2\)
& \(\varepsilon\) vs.\ \(\rho\)
& \(\begin{gathered}
   \tO\bigl(\sqrt n\,\rho^{-2}(\log Q)^{3/2}\\
   {}+n/(\sqrt m\rho)\bigr)
   \end{gathered}\)
& Thm.~\ref{thm:ug-tolerant} \\
\bottomrule
\end{tabularx}
\end{table}

\paragraph{Non-tolerant bipartiteness.}
Goldreich--Ron study the non-tolerant problem in the
bounded-degree model, and Kaufman--Krivelevich--Ron study it in the general
graph model~\cite{GR,KKR}.  The latter allow adjacency-pair queries and assume
\(m=\Omega(n)\).  Kaufman--Krivelevich--Ron also show that neighbor queries
alone require \(\Omega(\sqrt n)\) queries; their \(\tO(n^2/m)\) query bound
uses adjacency-pair queries~\cite[Section~5.2]{KKR}.  Our upper bounds remain
valid if such queries are available, but the dense lower bound in
Proposition~\ref{prop:query-lower-bounds} concerns our stated model, which excludes
that oracle.

More recently, Fei--Rubinfeld~\cite{FeiRubinfeld2026} show that
\(O_\rho(\sqrt n)\) walks of length \(O_\rho(\log n)\) suffice for
non-tolerant bipartiteness testing on simple graphs without isolated vertices
in the adjacency-list model.  Our result allows a positive completeness
threshold satisfying \(\varepsilon\log n\lesssim\rho^2/(1+\log(1/\rho))^2\)
and retains \(\tO_\rho(\sqrt n)\) queries when \(m=\Omega(n)\).

\paragraph{Tolerant bipartiteness.}
Tolerant MaxCut algorithms yield bipartiteness testers under structural
assumptions.
Peng--Yoshida obtain an \(\varepsilon\) vs.\ \(\rho\) guarantee on
conductance-\(\phi\) graphs when \(\rho\gtrsim\varepsilon/\phi^2\), with query
complexity \(\tO((\phi^2\rho)^{-1}
m^{1/2+O(\varepsilon/(\phi^2\rho))})\)~\cite{PengYoshida}.
Jha--Kumar handle
\(\varepsilon\) vs.\ \(1/2-\varepsilon\) on bounded-degree clusterable graphs
with \(\tO(n^{0.501+O(\varepsilon/\phi^2)})\) queries under the parameter
conditions specified below~\cite{JhaKumar}.
Our bipartiteness tester gives tolerance of \(c\lambda_\rho^2/\log n\)
for every fixed rejection
threshold \(\rho\) on arbitrary simple undirected graphs in the adjacency-list
model, without bounded-degree, expansion, or clusterability assumptions.

\paragraph{Tolerant Unique Games.}
The Unique Label Cover algorithm of Peng--Yoshida can likewise
be viewed as a tolerant tester on expanders~\cite{PengYoshida}.
Our Unique Games tester allows \(\varepsilon\log n\le c\lambda_\rho^2\)
on arbitrary simple undirected constraint graphs, for every alphabet size \(Q\).

\paragraph{Conventions and conditions in Table~\ref{tab:related-comparison}.}
The rows for our results assume simple undirected graphs and allow
isolated vertices.  Our Unique Games row allows arbitrary \(Q\), uses
queries returning entire permutations, and displays the dependence on
\(Q\) explicitly.  The notation \(\varepsilon\) vs.\ \(\rho\) denotes the
tolerant promise.  The bounds suppress polylogarithmic factors.
In the prior results, constants may additionally depend on alphabet size
or degree; \(c_Q,C_Q>0\) there denote constants depending only on \(Q\).

In the Goldreich--Ron row, distance is the minimum
fraction of the \(dn\) incidence-list entries that must change, where \(d\)
is the degree bound.  For bipartiteness this is \(2m\tau(G)/(dn)\), since
deleting an edge changes two entries.  All other rows use the minimum
number of violated edges or constraints divided by \(m\), including the
bounded-degree Unique Games row.  Thus the same numerical distance threshold
in the Goldreich--Ron row need not equal an edge-normalized threshold.

The row marked \(\dagger\) uses Theorem~16 of Jha--Kumar~\cite{JhaKumar}:
the number of clusters \(k\ge1\), degree bound \(d\ge1\), and
conductance parameter \(0<\phi<1\) are fixed, and the
maximum degree is at most \(d\).  The
graph admits a partition into \(k\) clusters of comparable sizes, each with
inner conductance at least \(\phi\) and outer conductance at most
\(\varepsilon\).  The parameter range is
\(\varepsilon\le10^{-9}\phi^2/(d^4\chi)\), where \(\chi>1\) is a
sufficiently large fixed constant.  The displayed exponent uses their
theorem's value \(1/2+100\cdot10^{-5}=0.501\).

\paragraph{Limits at constant tolerance.}
Chiplunkar--Kapralov--Khanna--Mousavifar--Peres~\cite[Theorem~6]{ChiplunkarKKMP}
show that, for every sufficiently small constant \(\eta>0\), estimating
the MaxCut value within a factor \(2-\eta\) requires
\(n^{1/2+\Omega(\eta/\log(1/\eta))}\) queries in the bounded-degree model.
They attribute this bound to Yoshida and give a different
proof~\cite[Remark~1]{ChiplunkarKKMP}.

Yoshida~\cite[Theorem~1.2]{YoshidaCSP} gives a lower bound for tolerant
testing of 2-XOR, the two-label case of Unique Games.  For every sufficiently
small constant \(\eta>0\), there are an integer \(d=d(\eta)\ge1\) and
a real constant \(\delta=\delta(\eta)>0\) such that distinguishing
instances of maximum degree at most \(d\) that are \(\eta\)-close to satisfiable from those that
are \((1/2-\eta)\)-far requires \(\Omega(n^{1/2+\delta})\) queries in the
bounded-degree CSP model.\footnote{Yoshida normalizes the minimum number of deleted
constraints by \(dn/2\).  On the \(d\)-regular instances used for the
lower bound, \(dn/2=m\), so this is the fraction of violated constraints.}
To transfer this bound to bipartiteness when the number of edges is known, replace
each equality constraint by a path of length four between its variables.
For an inequality constraint, use a path of length three and attach a
leaf to an internal vertex.  Use fresh internal vertices and leaves for
each constraint.  For fixed colors on the variables, a satisfied
constraint permits all four edges to be cut, and a violated constraint
leaves exactly one edge uncut in an optimal extension.  Thus the minimum
uncut fraction is one quarter of the minimum violated fraction.
The resulting graph is simple and has four edges per constraint in both
cases.  Its auxiliary vertices can be indexed by incidence-list positions,
leaving unused vertices isolated, so its size is linear and each graph
query takes \(O(d)\) constraint queries.  For any fixed
\(0<\varepsilon<\rho<1/10\), choosing the 2-XOR tolerance \(\eta>0\)
small enough gives \(\eta/4\le\varepsilon\) and
\((1/2-\eta)/4\ge\rho\), as needed.

\subsection{Technical overview}
\label{sec:proof-overview}

For bipartiteness, a walk from a seed predicts the color of its endpoint
through the parity of the number of edges traversed.  Two walks from the
same seed reaching the same vertex with opposite parities certify an odd
cycle.  The Goldreich--Ron tester rejects on finding such a pair, so
bipartite graphs are always accepted.  Soundness without rapid mixing is
more delicate: their analysis partitions the graph into pieces with
suitable local mixing properties and few edges between them.  If the
tester rarely finds an odd cycle, local two-colorings of these pieces can
be combined into a coloring with few violations~\cite{GR}.
Fei--Rubinfeld give a different analysis showing that walks of logarithmic
length suffice on simple graphs without isolated vertices in the
adjacency-list model~\cite{FeiRubinfeld2026}.

For tolerant testing, a graph close to bipartite may contain odd cycles,
so finding one is no longer a reason to reject; Peng--Yoshida explicitly
distinguish this difficulty from ordinary bipartiteness
testing~\cite[Section~1.1]{PengYoshida}.  Their expander tester compares
even- and odd-parity endpoint distributions.  Jha--Kumar show why this
approach does not directly extend to clusterable graphs: in an example
with two sparsely connected clusters, logarithmic-length walks at constant
tolerance make the two parity distributions close even though the graph
is close to bipartite~\cite[Section~3]{JhaKumar}.
Longer walks have more opportunities to traverse edges violated by a
nearly proper two-coloring, while sparse cuts obstruct rapid mixing.
Thus a soundness analysis for non-tolerant testing does not by itself
supply the completeness guarantee needed for tolerance.

Our tester measures the probability mass of conflicting predictions.
We describe it for Unique Games, where a walk predicts an endpoint label
by composing edge permutations along its path.  We first define the
ambiguity of these predictions and describe how the tester uses label
disagreements to decide whether to reject.  We then show why this rule
distinguishes close and far inputs even without rapid mixing.  Finally, we explain how
to implement the same decision rule on a sparsified graph within the
query bound.

\subsubsection{Setup and definition of ambiguity}

Assume \(m>0\), let \(V\) be the set of positive-degree vertices, and
keep \(n\) for the number of vertices in the original input.
Suppose a labeling \(x:V\to[Q]\) violates few edges.  A path from \(s\in V\)
with initial label \(x(s)\) that avoids those edges predicts \(x(v)\) at its
endpoint \(v\in V\).  We therefore ask how much path probability must be
discarded so that the remaining predictions agree on one label at each
endpoint.  Paths that traverse violated edges contribute at most their
total probability, even if they concentrate at a few endpoints.

We use boldface for random variables and ordinary type for fixed values,
including observed outcomes.  In particular, \(\rv{v}\) denotes a random
vertex and \(v\in V\) a possible value.
Fix a seed \(s\in V\) and an
initial label \(a\in[Q]\).  For a real parameter \(L\ge1\), a \emph{geometrically stopped walk} takes
\(\rv{T}_L\) steps of the lazy walk, where \(\rv{T}_L\) is independent of
the walk and has a geometric distribution on \(\{0,1,2,\ldots\}\).
Its expected number of steps is \(\E\rv{T}_L=L\).  We use geometric
stopping because, conditional on reaching a return time, the number of
remaining steps has the same geometric distribution.  A holding step preserves
the label.  Let \(p^a_{L,s}:V\times[Q]\to[0,1]\) give the probability of ending
at \(v\in V\) with label \(b\in[Q]\).  Define the endpoint
distribution \(u_{L,s}:V\to[0,1]\) by
\(u_{L,s}(v):=\sum_{b\in[Q]}p^a_{L,s}(v,b)\); it is independent
of \(a\).  For any nonnegative measure \(p\) on vertex-label pairs,
write \(\operatorname{Amb}(p)\) for the mass to discard to keep just one
most likely label at each vertex.  For this walk,
\[
 \operatorname{Amb}(p^a_{L,s})
 =\sum_{v\in V}\left(u_{L,s}(v)-\max_{b\in[Q]}p^a_{L,s}(v,b)\right).
\]
Since the tester does not know
the label of \(s\) in a nearly satisfying labeling, it uses the
function \(\mu_L:V\to[0,1]\) defined by
\(\mu_L(s):=\min_{a\in[Q]}\operatorname{Amb}(p^a_{L,s})\).

To estimate \(\operatorname{Amb}(p^a_{L,s})\) for a fixed initial label
\(a\), we compare two predictions at the same endpoint.
Fix \(v\) with \(u_{L,s}(v)>0\) and consider two independent
draws from the distribution of the final label of a walk from \((s,a)\),
conditional on ending at \(v\).  Choose a most likely label at \(v\) and let
\(\eta_v\in[0,1]\) be the conditional probability of a different label.
Given the first draw, the second matches it with probability at most
\(1-\eta_v\).  If the draws differ, at least one missed the chosen label.
The conditional disagreement probability is therefore between \(\eta_v\)
and \(2\eta_v\).  The average of \(\eta_v\) with weights \(u_{L,s}(v)\)
is exactly the discarded mass \(\operatorname{Amb}(p^a_{L,s})\).  Hence
the disagreement probability averaged over endpoints lies between
\(\operatorname{Amb}(p^a_{L,s})\) and twice its value.

\subsubsection{The tester's decision rule}

At a seed \(s\in V\) and mean walk length \(L\ge1\), one trial runs a
geometrically stopped walk from \(s\).  At its observed endpoint \(v\),
the walk predicts a label for each initial label \(a\in[Q]\).  For each
\(a\), compare that prediction with a fresh draw conditional on the same
\(v\) and the same \(a\).  Both predictions must concern the same vertex
because even a satisfying labeling can assign different labels to different
vertices.  One way to obtain the second draw is to repeat walks from \(s\)
until one ends at \(v\).  We give a cheaper implementation below.

If \(\mu_L(s)\le1/40\), the bound above gives some initial label with
disagreement probability at most \(1/20\), while if \(\mu_L(s)\ge1/4\), every initial label has
disagreement probability at least \(1/4\).  The tester reports HIGH when
the observed disagreement fraction is at least \(1/8\) for every initial
label, and LOW otherwise.  For \(0<\delta<1/2\), Hoeffding's inequality
and a union bound over the labels give error at most \(\delta\) in either
case using \(O(\log(Q/\delta))\) independent trials.

To decide whether to reject the instance, the tester counts HIGH reports
at several walk lengths.  The soundness argument below shows that in a
far instance there is a scale \(r\in(0,1]\) with mean walk length \(L_r\)
at which a degree-proportional seed \(s\) satisfies \(\mu_{L_r}(s)\ge1/4\)
with probability at least \(r/4\).  Since that scale is unknown, the tester tries dyadic scales
from \(\Theta(\lambda_\rho)\) up to \(1\), including \(1\) itself, where
\(\lambda_\rho=\rho/(1+\log(1/\rho))\).  At each scale \(r\), it samples
\(\tO(1/r)\) seeds approximately in proportion to degree and tests each
at mean length \(L_r=\Theta(r\lambda_\rho^{-2}\log n)\) with error probability
\(\delta=r/400\).  It rejects if the observed fraction of HIGH reports
reaches \(r/32\) at any scale and otherwise accepts.

\subsubsection{Correctness on the input graph}

To prove completeness, we bound the probability of a HIGH report at
every scale in a close instance.  To prove soundness, we find one scale
in a far instance at which that probability exceeds the rejection threshold.

\paragraph{Completeness.}
If \(x\) violates at most \(\varepsilon m\) edges, a walk starting with
label \(x(s)\) can disagree with \(x\) at its endpoint only if it crosses
a violated edge.  The probability of such a crossing therefore bounds
\(\operatorname{Amb}(p^{x(s)}_{L,s})\).
Draw the starting vertex \(\rv{s}\) from the stationary distribution
\(\pi:V\to(0,1]\) given by \(\pi(v):=\deg_G(v)/(2m)\).  Then
each step traverses a violated edge with
probability \(O(\varepsilon)\).  Averaging up to the geometric stopping time
gives the completeness bound of Lemma~\ref{lem:ug-completeness},
\[
 \E_{\rv{s}\sim\pi}\mu_L(\rv{s})=O(\varepsilon L).
\]
By Markov's inequality, the probability that \(\mu_{L_r}(\rv{s})>1/40\)
is \(O(\varepsilon L_r)\), even allowing for the pointwise approximation
to degree-proportional sampling.  Each seed \(s\) with
\(\mu_{L_r}(s)\le1/40\) reports HIGH with probability at most \(r/400\).
Since \(\varepsilon L_r/r=\Theta(\varepsilon\log n/\lambda_\rho^2)\),
the theorem's tolerance condition makes the total HIGH probability at
most \(r/128\) at every scale.  This is where the inverse-logarithmic
tolerance enters the analysis.

\paragraph{Soundness.}
A far instance can contain large regions with nearly consistent labels.
To find seeds with large ambiguity, the proof progressively removes
such pieces.  Write
\(R\subseteq V\) for the nonempty set of remaining positive-degree
vertices.  To relate this removal process to the original walks, we keep
their excursions through removed vertices: the \emph{trace chain} records
successive visits to \(R\), transporting labels along each entire
excursion.  Its stationary distribution \(\pi_R:R\to(0,1]\) is
defined by \(\pi_R(v):=\pi(v)/\pi(R)\) for \(v\in R\).
The set \(R\) and its trace chain are used only in this soundness proof.
They allow us to identify seeds at which the tester's ordinary walks will
produce conflicting labels.

For a partial labeling \(\ell:S\to[Q]\), with
\(\emptyset\ne S\subseteq R\), a trace step \emph{fails} if it leaves \(S\)
or disagrees with \(\ell\) at the return vertex.  Its failure probability
is measured from \(\pi_R\) conditioned on \(S\).  A small failure
probability implies few internal violations and few edges between
\(S\) and \(R\setminus S\), relative to the total degree of \(S\),
because an ordinary step across either kind of edge, starting in \(S\),
already returns to \(R\) and fails \(\ell\).  We can therefore remove \(S\) and retain
\(\ell\) as its labeling.  When these partial labelings are combined,
internal violations are charged to their piece and crossing edges to
the earlier removed piece.

We repeat this removal until either no partial labeling has sufficiently
small trace failure probability, or the remaining mass is small enough to
discard.  In the latter case, we could combine the saved partial labelings
and label the remainder arbitrarily.  We choose the thresholds so that
this would give a labeling with violation rate less than \(\rho\),
contradicting \(\rho\)-farness.  Thus removal must stop before the remainder
becomes negligible, because no partial labeling meets the removal threshold.

Quantitatively, we remove a partial labeling when its trace failure
probability is below a small constant times \(\lambda_\rho/\pi(R)\),
and discard the remainder if its mass falls below a small constant times
\(\lambda_\rho\).  If the remaining mass fell below this threshold,
the total violation rate would be at most a small constant times
\(\lambda_\rho(1+\log(1/\lambda_\rho))=O(\rho)\), which is less than
\(\rho\) when the constants are chosen sufficiently small.  The logarithm
comes from summing the removed mass divided by the mass remaining before
each removal.  Lemma~\ref{lem:ug-multiscale-trace-residual} thus gives a residual \(R\) with
\(\pi(R)=\Omega(\lambda_\rho)\) in which every partial labeling has trace
failure probability \(\Omega(\lambda_\rho/\pi(R))\).

We next use this residual set \(R\) to show that trace walks from many
seeds produce conflicting labels.  Let \(\theta\in(0,1]\) be a common
lower bound on the trace failure probabilities of all partial labelings in \(R\),
each measured from \(\pi_R\) conditioned on the labeling's domain.
The removal argument allows us to take \(\theta=\Omega(\lambda_\rho/\pi(R))\).
This lower bound concerns assignments of one label per vertex; it need
not make the walk on vertex-label pairs mix.  To force constant ambiguity,
it is enough to make the total excess of the most likely labels above
half the endpoint mass small.  We therefore track only this excess.
For an integer \(t\ge0\), let
\(p_t^a:R\times[Q]\to[0,1]\) be the distribution after \(t\) trace steps
from \((s,a)\in R\times[Q]\).  Define \(u_t:R\to[0,1]\) by
\(u_t(v):=\sum_{b\in[Q]}p_t^a(v,b)\), and define
\(g_t:R\times[Q]\to\mathbb R_{\ge0}\) by
\[
 g_t(v,b):=\frac{(2p_t^a(v,b)-u_t(v))_+}{\pi_R(v)},
 \qquad y_+:=\max\{y,0\}\quad(y\in\mathbb R).
\]
Thus \(g_t(v,b)\) is \(2/\pi_R(v)\) times the mass of label \(b\) above
the threshold \(u_t(v)/2\).  At most one label per vertex exceeds this
threshold.  Each nonempty positive superlevel set of \(g_t^2\) is therefore
the graph of a partial labeling \(\ell:S\to[Q]\).  The flow leaving that
graph in one trace step is \(\pi_R(S)\) times the failure probability of
\(\ell\), hence at least \(\theta\pi_R(S)\).

Integrating this flow bound over the superlevel sets gives a lower bound
of \((\theta^2/2)\|g_t\|_*^2\) on the Dirichlet energy of \(g_t\) for the
trace walk on vertex-label pairs, where
\(\|g\|_*^2:=\sum_{v\in R,b\in[Q]}\pi_R(v)g(v,b)^2\).
If \(\widehat K_R\) denotes the transition kernel of this walk,
reversibility and the definition of \(g_t\) give
\(g_{t+1}\le\widehat K_R g_t\) pointwise.
Laziness and the lower bound on the Dirichlet energy then give
\[
 \|g_{t+1}\|_*^2\le(1-\theta^2/2)\|g_t\|_*^2.
\]
Replacing each summand in \(1-2\operatorname{Amb}(p_t^a)\) by its
positive part gives the first inequality below.  The second follows
by Cauchy--Schwarz, since at most one \(g_t(v,b)\) is positive at each \(v\):
\[
 1-2\operatorname{Amb}(p_t^a)
 =\sum_v\bigl(2\max_b p_t^a(v,b)-u_t(v)\bigr)
 \le\sum_{v,b}\pi_R(v)g_t(v,b)\le\|g_t\|_*.
\]
Iterating the contraction from \(\|g_0\|_*=1/\sqrt{\pi_R(s)}\)
therefore gives the bound of Lemma~\ref{lem:ug-trace-time-majority-excess},
\[
 \operatorname{Amb}(p_t^a)
 \ge\frac{1}{2}\left(1-\frac{e^{-\theta^2t/4}}{\sqrt{\pi_R(s)}}\right).
\]
The bound depends on \(s\) but not on its initial label, which is what
allows us to control the minimum over all \(Q\) labels.

To transfer this ambiguity bound to the tester's walks, we must relate
the number of returns to \(R\) to the number of ordinary steps.  Let \(\rv{H}_k\in\mathbb Z_{\ge1}\) be
the ordinary time of the \(k\)-th return, for an integer \(k\ge1\).
Write \(\E_s\) for expectation when the walk starts at \(s\).  If geometric stopping
has not occurred by \(\rv{H}_k\), memorylessness gives a fresh geometrically
stopped walk from that point.  Permuting labels preserves ambiguity, and merging path
mass at an endpoint can only increase it.  Continuing the surviving
paths therefore cannot decrease their ambiguity; the loss comes from
paths stopped before \(\rv{H}_k\).  Lemma~\ref{lem:trace-geometric-comparison}
therefore gives
\[
 \operatorname{Amb}(p^a_{L,s})
 \ge\operatorname{Amb}(p_k^a)-\frac{\E_s\rv{H}_k}{L+1}.
\]
By Kac's return-time formula (Lemma~\ref{lem:stationary-return-time}),
\(\E_{\rv{s}\sim\pi_R}\rv{H}_k\le k/\pi(R)\).  Averaging also controls the seed
dependence in the lower bound on \(\operatorname{Amb}(p_k^a)\), since
\(\sum_{s\in R}\sqrt{\pi_R(s)}\le\sqrt n\).  Thus
\(k=\Theta(\theta^{-2}\log n)\) and \(L=\Theta(k/\pi(R))\), with large
enough constants, give ambiguity at least \(1/4\) for every initial label
on a set of seeds of \(\pi\)-probability at least \(\pi(R)/2\)
(Lemma~\ref{lem:ug-trace-to-geometric}).

Since \(\theta=\Omega(\lambda_\rho/\pi(R))\), the bound on \(k\) gives
\(k/\pi(R)=O(\pi(R)\lambda_\rho^{-2}\log n)\).  The tester does not know
\(\pi(R)\), but its dyadic scales range from \(\Theta(\lambda_\rho)\) up
to \(1\) and include an \(r\) with \(\pi(R)\le r\le2\pi(R)\).
Choosing the constant in \(L_r\) sufficiently large therefore gives
\(\mu_{L_r}(s)\ge1/4\) on a set of seeds of stationary probability
at least \(\pi(R)/2\ge r/4\)
(Lemma~\ref{lem:ug-multiscale-ambiguity}).
At this scale each seed \(s\) with \(\mu_{L_r}(s)\ge1/4\) reports HIGH with
probability at least \(1-r/400\).  Sampling seeds within a factor
\(1\pm 1/10\) of their stationary probabilities thus gives a HIGH report
with probability at least \(r/8\).  The cutoff \(r/32\) lies between this
probability and the completeness bound \(r/128\).  With \(\tO(1/r)\)
independent seeds per scale, a Chernoff bound controls the fraction of
HIGH reports, with a union bound over scales in the close case.

\subsubsection{Implementation and query complexity}

The decision rule requires two predictions for each initial label at the
same endpoint.  We now show how to sample the second label using
adjacency-list queries, first on the input graph and then on a sparsifier
that reduces the cost on dense inputs.

\paragraph{Sampling a label conditional on the endpoint.}
Two independent walks from \(s\) both end at \(v\) with probability
\(u_{L,s}(v)^2\), whereas the disagreement test needs to weight \(v\)
by \(u_{L,s}(v)\).  We could keep the endpoint \(v\) of the first walk
and run fresh walks until one ends there.  This gives the desired conditional
label for every initial label, but takes \(1/u_{L,s}(v)\) additional walks
in expectation.  We instead use the forward-push/reverse-walk decomposition
of bidirectional PageRank~\cite{LofgrenBanerjeeGoel}.  We calculate part
of the endpoint-label distribution by exploring from \(s\), then use
reversibility to sample from the remaining part via walks starting at
\(v\).  This gives exact conditional samples without waiting for another
walk from \(s\) to reach \(v\).

The amount of exploration controls a tradeoff between preprocessing
and sampling.  For a real parameter \(h>0\), exposing and caching the
neighborhoods needed for preprocessing costs \(O(1+m/h)\) queries,
following local-push query analyses~\cite[Section~1]{CuiWeiYang2026}.
In one trial the first walk draws an endpoint \(\rv{v}\sim u_{L,s}\).
The vertex trajectory is independent of the initial label, and each
edge query returns the entire permutation.  All \(Q\) initial labels
therefore use the same forward exploration and the same reverse walks
from \(\rv{v}\), with one acceptance decision per attempt.  For every
seed \(s\), the expected
number of reverse walks to draw all \(Q\) second labels in this trial
is at most \(h\), averaged over \(\rv{v}\) and the sampling randomness.
Including the first walk, one trial therefore costs \(O(L(1+h))\)
expected queries.  On the input graph the dependence on \(Q\) in the
query bound comes from running \(O(\log(Q/\delta))\) trials.
Choosing \(h=\sqrt{m/L}\) balances preprocessing and sampling, giving
\(O((\sqrt{Lm}+L)\log(Q/\delta))\) expected queries to test one seed.
Algorithm~\ref{alg:ambiguity-test} gives the procedure on the sparsifier
and Proposition~\ref{sparse:prop:shared-ambiguity} analyzes it.
The dependence on \(m\) is still too large on dense graphs.

\paragraph{Reducing the cost on dense graphs.}
The \(\sqrt m\) dependence comes from exposing neighborhoods.  We reduce
this cost by testing a sparse weighted instance.  After discarding edges
incident to low-degree vertices when \(p<1\), we retain each remaining
edge independently with probability
\[
 p:=\min\{1,\Theta(n\log(2nQ)/(\rho m))\}\in(0,1].
\]
Write \(H\) for the sampled graph and \(k_H(v)\) for the number of
retained edges at \(v\).  With probability at least \(99/100\),
\(k_H(v)\le2p\deg_G(v)\) at every vertex.  On this event the walk on
\(H\) is lazy and reversible with the original stationary distribution
\(\pi\), so we can still sample seeds using edges of the input graph.
By Lemma~\ref{sparse:lem:sparsification-promise}, with probability at
least \(98/100\) the degree bound holds and \(H\) satisfies the following
guarantee for the given input.  The minimum, over labelings, of the
probability that a stationary step on \(H\) violates a constraint is at
most \(25\varepsilon\) when \(\tau(\mathcal U)\le\varepsilon\), and at
least \(\rho/16\) when \(\tau(\mathcal U)\ge\rho\).

To obtain a sublinear query bound, we reveal neighborhoods and simulate
walks only as needed, keeping all operations consistent with one sampled
graph.  For every \(H\) satisfying \(k_H(v)\le2p\deg_G(v)\) at all
vertices, our local simulation reduces the preprocessing cost
from \(O(1+m/h)\) to \(O(1+pm/h)\) expected queries, while each walk
still costs \(O(L)\) expected queries, conditional on the sampled graph
and preceding history.  Lemma~\ref{sparse:lem:deferred-oracle} gives the conditional costs
of a walk and a neighborhood exposure.
Balancing preprocessing and sampling with \(h=\sqrt{pm/L}\) now gives
\[
 O\bigl((\sqrt{Lpm}+L)\log(Q/\delta)\bigr)
\]
expected queries per seed (Proposition~\ref{sparse:prop:shared-ambiguity}).  Since \(pm=O(n\log(nQ)/\rho)\), the
cost of estimating ambiguity depends on the number of vertices rather
than the number of input edges.

\paragraph{Query complexity.}
Multiplying the cost of one ambiguity test on the shared sparsifier by the
\(\tO(1/r)\) seeds at scale \(r\), with \(L=L_r\) and per-call error
\(\delta=r/400\), gives
\[
 \tO\left(\left[\frac{\sqrt{pm}}{\lambda_\rho\sqrt r}
                       +\frac{1}{\lambda_\rho^2}\right]\log Q\right).
\]
The smallest scale is \(\Theta(\lambda_\rho)\), so summing over scales
contributes a factor \(O(\lambda_\rho^{-1/2})\) to the first term.
Substituting the bound on \(pm\) into \(\sqrt{pm}\) contributes
\(\sqrt{\log Q}\) to the first term.  Repeating the disagreement
trial at each seed contributes \(\log Q\) to both terms.
To draw a seed, the tester samples an
almost-uniform edge of the input graph using Lemma~\ref{lem:almost-uniform-edge-sampler}
and chooses an endpoint uniformly.  This costs \(\tO(n/\sqrt m)\) per seed.
Summing over the \(\tO(1/\lambda_\rho)\) seeds gives the remaining
\(\tO(n/(\sqrt m\lambda_\rho))\) term.
For fixed \(\rho\) and \(Q\), these are the two contributions
\(\tO(\sqrt n)\) and \(\tO(n/\sqrt m)\).

Finally, let \(B\) be the upper bound in~\eqref{eq:ug-query-budget}
on the expected number of queries before truncation, conditional on
any \(H\) satisfying the retained-degree bound.  The tester stops and
accepts when the query count reaches \(\lceil100B\rceil\) or when it
finds a vertex with \(k_H(v)>2p\deg_G(v)\).  For each \(H\) satisfying
the degree bound, Markov's inequality bounds the probability of reaching
\(\lceil100B\rceil\) by \(1/100\).

\subsection{Related work}
\label{sec:relation-prior-work}

\paragraph{Tolerant bipartiteness in dense graphs.}
In the dense graph model, algorithms have adjacency-matrix access and distance is
normalized by \(n^2\); inspecting an induced subgraph on \(q\) sampled vertices
uses \(O(q^2)\) matrix queries.  The sampling framework of Alon, Fernandez de la
Vega, Kannan, and Karpinski approximates dense MaxCut values and, more generally,
dense Boolean MAX-\(r\)CSP values for fixed integer arity \(r\ge1\), to
additive error \(\eta>0\) in normalized value from
\(q=O(\log(1/\eta)/\eta^4)\) sampled vertices~\cite{AFKK}.  The recent tolerant
bipartiteness tester of Ghosh,
Mishra, Raychaudhury, and Sen distinguishes \(\varepsilon\)-close from
\((2+\Omega(1))\varepsilon\)-far from bipartite using
\(\tO(1/\varepsilon^3)\) adjacency-matrix queries~\cite{GhoshMRS}.

\paragraph{Expansion and clusterability testing.}
The structural assumptions in several comparison works are also part of a
larger literature on random-walk-based testing of expansion and cluster
structure.  Expansion testing in bounded-degree graphs was studied by
Czumaj--Sohler and Kale--Seshadhri~\cite{CzumajSohlerExpansion,KaleSeshadhri}.
For cluster structure, Czumaj--Peng--Sohler gave a sublinear
tester~\cite{CzumajPengSohler}, and
Chiplunkar--Kapralov--Khanna--Mousavifar--Peres gave clusterability testers
and lower bounds~\cite{ChiplunkarKKMP}.
More recent work gives sublinear-time oracles for determining a vertex's
cluster under clusterability assumptions, with improvements in query
time~\cite{KapralovKochetkovaWrzosKaminska2026} and memory
usage~\cite{ShenZhuPengHuang2026}.
Recent expansion testers also address graphs whose connected components
are all expanders~\cite{DinurGoldreich2025} and the existence of a large
induced expander when the rest of the graph is
arbitrary~\cite{PengShen2026}.

\paragraph{MaxCut and Unique Games approximation.}
Related dense MaxCut approximation schemes include the greedy algorithm of
Mathieu and Schudy~\cite{MathieuSchudy}.  The Goemans--Williamson SDP algorithm
and the Khot--Kindler--Mossel--O'Donnell hardness result are the classical
landmarks for MaxCut approximation~\cite{GoemansWilliamson,KKMO}.  On the Unique
Games side, expander constraint graphs and spectral algorithms were studied by
Arora--Khot--Kolla--Steurer--Tulsiani--Vishnoi and by
Kolla~\cite{AKKSTV,Kolla}.  Gupta and Talwar gave an
\(O(\log n)\)-approximation for the
minimization version of Unique Games~\cite{GuptaTalwar}.  Their algorithm
maps an instance with minimum violated fraction \(\tau(\mathcal U)=\varepsilon\) to a
labeling with \(O(\varepsilon m\log n)\) violated edges.  Thus their
approximation guarantee separates the two promises when
\(\varepsilon\log n\) is a sufficiently small constant multiple of \(\rho\).
Trevisan's SDP algorithm for
Unique Games gives a labeling satisfying a constant fraction of constraints
when the maximum satisfiable fraction \(1-\tau(\mathcal U)\) is
\(1-O(1/\log n)\)~\cite{TrevisanUG}.  For fixed
\(\rho\), our tester also allows \(\varepsilon\le c\lambda_\rho^2/\log n\).

\subsection{Organization}
\label{sec:organization}

Section~\ref{sec:preliminaries} defines Unique Games instances and the adjacency-list
query model, then gives the walk conventions and sampling bounds used by the tester.
Section~\ref{sec:tolerant-unique-games} defines ambiguity and proves its
completeness and soundness bounds, then implements the ambiguity test on
a locally simulated sparsifier.  It concludes with the full tester and
the proof of Theorem~\ref{thm:ug-tolerant}, which also implies
Theorem~\ref{tol:thm:main}.
Section~\ref{sec:query-lower-bounds} proves
Proposition~\ref{prop:query-lower-bounds}, the
\(\Omega(\sqrt n+n/\sqrt m)\) query lower bound.

\section{Preliminaries}
\label{sec:preliminaries}

We define Unique Games instances and the queries available to the tester, then
collect the return-time and edge-sampling bounds used in its analysis.

\subsection{Unique Games and adjacency-list queries}
\label{sec:oracle-models}

\begin{definition}[Unique Games instance]
An input graph is a finite undirected simple graph
\[
    G_{\mathrm{in}}=(V_{\mathrm{in}},E),
    \qquad
    n:=|V_{\mathrm{in}}|,
    \qquad
    m:=|E|.
\]
Parallel edges and self-loops are excluded; isolated vertices are allowed.

For an integer \(Q\ge2\), write \([Q]=\{1,\ldots,Q\}\) and let \(S_Q\)
be the group of permutations of \([Q]\).  A Unique Games instance
\(\mathcal U\) on \(G_{\mathrm{in}}\) has one constraint per edge: for each
orientation \(uv\) of an edge, the constraint is a permutation
\(\pi_{uv}\in S_Q\), with \(\pi_{vu}=\pi_{uv}^{-1}\) promised by the input.
A labeling \(x:V_{\mathrm{in}}\to[Q]\) satisfies \(\{u,v\}\) if
\(x(v)=\pi_{uv}(x(u))\).  For \(m>0\), the minimum violated fraction is
\[
    \tau(\mathcal U):=\min_{x:V_{\mathrm{in}}\to[Q]}
    \frac{|\{\{u,v\}\in E:x(v)\ne\pi_{uv}(x(u))\}|}{m}.
\]
\end{definition}

\begin{definition}[adjacency-list model]
\label{def:adjacency-list-model}
The algorithm is given the number of edges \(m\).  The oracle supports the following queries.
\begin{itemize}
    \item A \emph{uniform-vertex query} returns a vertex sampled uniformly from
    \(V_{\mathrm{in}}\).
    \item A \emph{degree query} on \(v\in V_{\mathrm{in}}\) returns \(\deg(v)\).
    \item A \emph{neighbor query} on \((v,i)\), where \(v\in V_{\mathrm{in}}\) and
    \(i\in\{1,\ldots,\deg(v)\}\), returns
    the \(i\)-th neighbor of \(v\) in a fixed ordering of its adjacency list.
\end{itemize}
For a Unique Games input, the neighbor query also returns the entire
permutation from the queried vertex to the returned neighbor, specified by
its values on \([Q]\).  We call this a \emph{constrained neighbor query}.  Each
query costs one, independently of \(Q\); local computation and storage are
not charged.  The model does not include adjacency-pair queries.
\end{definition}

\paragraph{Conventions.}
Unless explicitly stated otherwise, query complexity is the maximum number
of queries over all inputs with the given parameters and all random choices
of the algorithm and the oracle.  We use expected costs to analyze internal
subroutines; the full tester enforces a deterministic bound on the total
number of queries, as specified in Subsection~\ref{subsubsec:ug-algorithm}.

When \(m=0\), the instance has no violated edges or constraints and the tester
accepts immediately.  In every nontrivial analysis below we restrict to the
set \(V\subseteq V_{\mathrm{in}}\) of positive-degree vertices and write
\(G=(V,E)\).  This restriction is only for analysis: the algorithm does not
search for or delete isolated vertices.  Uniform-vertex queries still sample
from \(V_{\mathrm{in}}\), and all query bounds use the original
\(n=|V_{\mathrm{in}}|\).  Every labeling of \(V\) extends arbitrarily to
the isolated vertices with the same violated edges, so \(\tau(\mathcal U)\)
is unchanged.

The simplicity assumption applies to the input graph.  Auxiliary objects,
including trace chains and label lifts, are treated as weighted graphs or
Markov kernels and may have parallel edges or transitions and self-loops.
Parallel contributions to the same matrix entry are summed.  General
Markov-chain and weighted-walk statements below retain their stated domains.

\subsection{Finite Markov chains and return times}
\label{sec:common-markov-chain-facts}

On a finite state space \(V\), a \emph{Markov kernel} \(P:V\times V\to[0,1]\)
specifies the probability \(P(u,v)\) of moving from \(u\) to \(v\) in one
step, with \(\sum_{v\in V}P(u,v)=1\) for each \(u\in V\).
A distribution \(\pi\) on \(V\) is \emph{stationary} if \(\pi P=\pi\).
The chain is \emph{reversible} with respect to \(\pi\) if
\(\pi(u)P(u,v)=\pi(v)P(v,u)\) for all \(u,v\in V\).

Two states \emph{communicate} if each can reach the other by a path of
positive-probability transitions.  The resulting equivalence classes are
the \emph{communicating classes}; the chain is \emph{irreducible} if it has only
one such class.  A class is \emph{closed} if the walk cannot leave it.
For a reversible chain with \(\pi(v)>0\) at every state, every
positive-probability transition can be reversed.  The classes are therefore
the connected components of the undirected graph of possible moves, and
each is closed.

Unless specified otherwise, the \emph{lazy random walk} on \(G\) stays put with
probability \(1/2\), and with probability \(1/2\) traverses a uniformly random
incident edge.  Its stationary distribution \(\pi:V\to[0,1]\) is defined by
\[
    \pi(v):=\frac{\deg(v)}{2m}.
\]
The lazy holding step is artificial and is not an input self-loop.

The following return-time bound lets us compare a prescribed number of
trace steps with an ordinary walk, even when the kernel is reducible.

\begin{lemma}[stationary mean return time]
\label{lem:stationary-return-time}
Let \(P\) be a finite reversible Markov kernel on \(V\) with a positive
stationary distribution \(\pi:V\to(0,1]\).  Irreducibility is not required.
Let \(R\subseteq V\) be nonempty, and define the conditional distribution
\(\pi_R:R\to(0,1]\) by \(\pi_R(v):=\pi(v)/\pi(R)\).
For each integer \(k\ge1\) and a \(P\)-walk with initial vertex \(\rv{s}\sim\pi_R\), let \(\rv{H}_k\)
be the time of its \(k\)-th visit to \(R\) after time zero.  Then
\[
    \mathbb E_{\rv{s}\sim\pi_R}\rv{H}_k\le \frac{k}{\pi(R)}.
\]
\end{lemma}

\begin{proof}
Let \(\mathcal C(R)\) be the collection of closed communicating classes
\(C\) of \(P\) with \(R\cap C\neq\emptyset\).  Conditional on \(\rv{s}\in C\), the walk never leaves
\(C\), the relevant return set is \(R\cap C\), and the conditional start
distribution is \(\pi_{R\cap C}\).  Kac's formula~\cite{Kac}, applied inside
\(C\), gives mean return time \(\pi(C)/\pi(R\cap C)\) between successive visits
to \(R\cap C\).  Therefore
\[
    \mathbb E[\rv{H}_k\mid \rv{s}\in C]
    =
    k\,\frac{\pi(C)}{\pi(R\cap C)}.
\]
Averaging over the starting communicating class gives
\[
\begin{aligned}
    \mathbb E_{\rv{s}\sim\pi_R}\rv{H}_k
    &=
    \sum_{C\in\mathcal C(R)}
    \frac{\pi(R\cap C)}{\pi(R)}
    \cdot
    k\frac{\pi(C)}{\pi(R\cap C)} \\
    &=
    \frac{k}{\pi(R)}
    \sum_{C\in\mathcal C(R)}\pi(C)
    \le
    \frac{k}{\pi(R)}.
\end{aligned}
\]
\end{proof}

\subsection{Sampling edges almost uniformly in the adjacency-list model}
\label{sec:edge-sampling}

The following theorem of Eden and Rosenbaum~\cite{EdenRosenbaum} lets the
algorithms sample an edge of the simple input graph almost uniformly using
only the adjacency-list queries of Definition~\ref{def:adjacency-list-model}.
The tester calls the edge sampler on \(G_{\mathrm{in}}\), including its
isolated vertices, and the cost uses the original \(n\).  To draw a seed,
the tester then chooses an endpoint of the returned edge uniformly.
The resulting vertex distribution is pointwise close to \(\pi\).

\begin{lemma}[pointwise almost-uniform sampler for edges]
\label{lem:almost-uniform-edge-sampler}
Assume \(m>0\) is known up to a constant factor.  For every accuracy parameter
\(\xi\in(0,1)\), there is a randomized algorithm that returns an edge from a
distribution \(\widetilde\mu_E:E\to[0,1]\) satisfying
\[
    \frac{1-\xi}{m}
    \le
    \widetilde\mu_E(e)
    \le
    \frac{1+\xi}{m}
    \qquad\text{for every }e\in E.
\]
Its expected query complexity is
\[
    \tO\!\left(\frac{n}{\sqrt{\xi m}}\right).
\]
\end{lemma}

\section{Tolerant Testing for Unique Games}
\label{sec:tolerant-unique-games}

This section proves Theorem~\ref{thm:ug-tolerant} for arbitrary alphabet
sizes and permutation constraints.  We first define ambiguity and show how
it separates nearly satisfiable instances from far ones.  We then construct
a locally simulated sparsifier and test its ambiguity by sampling labels
conditional on walk endpoints.  Sharing exploration across labels bounds
the estimation cost for every seed.  With these ingredients in place, we
state the full tester and prove its correctness and query bound.
Bipartiteness follows by taking two labels and transposition constraints.

\subsection{Ambiguity}
\label{subsec:ug-setup}

We define ambiguity for a weighted walk so that the same statistic applies
both to the input and to its sparsifier.  Allowing extra holding probability
lets the sparsifier preserve the stationary distribution of the input walk.

\begin{definition}[weighted constraint kernel]
\label{def:ug-weighted-kernel}
Let \(V\) be a nonempty set of at most \(n\) states.  For distinct
\(u,v\in V\), let \(w_{uv}=w_{vu}\ge0\) be real weights, and for each
\(v\in V\), let \(D_v>0\) be a real number satisfying
\(D_v\ge2\sum_{u\ne v}w_{uv}\).  Each positive-weight pair carries a
permutation constraint \(\pi_{uv}\in S_Q\) with \(\pi_{vu}=\pi_{uv}^{-1}\).
Write \(\mathcal H:=(V,w,D,\{\pi_{uv}\})\).  Define the Markov kernel
\(M:V\times V\to[0,1]\), the volume function
\(\vol:2^V\to\mathbb R_{\ge0}\), and the stationary distribution
\(\pi:V\to(0,1]\) by
\[
 \begin{gathered}
 M(u,v):=\frac{w_{uv}}{D_u}\quad(u\ne v),\qquad
 M(u,u):=1-\frac{\sum_{v\in V\setminus\{u\}}w_{uv}}{D_u},\\
 \vol(S):=\sum_{v\in S}D_v\quad(S\subseteq V),\qquad
 \pi(v):=\frac{D_v}{\vol(V)}\quad(v\in V).
 \end{gathered}
\]
Holding transitions carry the identity permutation.  For a labeling
\(x:V\to[Q]\), let \(B_x\subseteq\binom{V}{2}\) be its set of violated
positive-weight pairs.  For \(B\subseteq\binom{V}{2}\), define its weight
\(w(B)\) and the minimum violation probability \(\tau_w(\mathcal H)\) by
\begin{equation}
 \tau_w(\mathcal H)
 :=\min_{x:V\to[Q]}\frac{2w(B_x)}{\vol(V)},
 \qquad w(B):=\sum_{uv\in B}w_{uv}.
 \label{eq:ug-weighted-distance}
\end{equation}
For disjoint \(S,T\subseteq V\), write
\(w(S,T):=\sum_{u\in S,v\in T}w_{uv}\).
\end{definition}

The kernel \(M\) is lazy and reversible, possibly with absorbing states.
\(\tau_w(\mathcal H)\) is the minimum, over labelings \(x\), of the
probability that one stationary step traverses an edge violated by \(x\).  For the ordinary
lazy input walk, take \(w_{uv}:=1\) on input edges and \(D_v:=2\deg_G(v)\);
then \(\tau_w=\tau(\mathcal U)/2\).
The ambiguity definitions and the completeness and soundness lemmas below
concern an arbitrary weighted constraint kernel \(\mathcal H\).  In particular,
volume includes identity holding mass.  The simplicity restriction remains
an assumption on the input oracle, not on these analytic kernels.

For a nonnegative integer \(t\) and a walk path on \(V\),
\[
    P=u_0u_1\cdots u_t,
\]
define its \emph{label-transport permutation} \(\Pi_P\in S_Q\) by
\[
    \Pi_P
    :=
    \pi_{u_{t-1}u_t}\cdots \pi_{u_0u_1}.
\]
For a lazy holding step \(u_i=u_{i+1}\), the corresponding factor in the
label transport is the identity permutation.  Thus the product defining
\(\Pi_P\) is understood after inserting identity factors for all holding
steps.  We also set \(\Pi_P:=\mathrm{id}\) when \(t=0\).
Thus, if the label at \(u_0\) is \(a\), then the label predicted at \(u_t\)
along \(P\) is \(\Pi_P(a)\).

For the walk with kernel \(M\) from Definition~\ref{def:ug-weighted-kernel},
let \(\rv{v}_t\in V\) be its vertex at integer time \(t\ge0\) and
\(\rv{P}_t:=(\rv{v}_0,\ldots,\rv{v}_t)\in V^{t+1}\) its path through time \(t\).
For \(L\in[1,\infty)\), let \(\rv{T}_L\) be a nonnegative-integer-valued
geometric random variable, independent of the walk, with
\[
    \Pr[\rv{T}_L=t]
    =
    \frac{1}{L+1}\left(\frac{L}{L+1}\right)^t,
    \qquad t=0,1,2,\ldots.
\]
The geometrically stopped walk \(\rv{P}_{\rv{T}_L}\) takes \(\rv{T}_L\)
steps, including holding steps, and \(\E\rv{T}_L=L\).

For a seed \(s\in V\) and an initial label \(a\in[Q]\), define the joint endpoint distribution \(p^a_{L,s}:V\times[Q]\to[0,1]\) by
\[
    p^a_{L,s}(v,b)
    :=
    \Pr[\rv{v}_{\rv{T}_L}=v,\ \Pi_{\rv{P}_{\rv{T}_L}}(a)=b\mid \rv{v}_0=s]\qquad(v\in V,\ b\in[Q]).
\]
Define the distribution \(u_{L,s}:V\to[0,1]\) of the endpoint vertex by
\[
    u_{L,s}(v)
    :=
    \sum_{b=1}^Q p^a_{L,s}(v,b),
\]
which is independent of \(a\).

For a finite set \(X\) and a nonnegative measure \(p:X\times[Q]\to\mathbb R_{\ge0}\),
define its \emph{ambiguity} as the minimum mass one must discard to keep
at most one label at each endpoint:
\[
 \operatorname{Amb}(p):=\sum_{v\in X}\left(\sum_{b\in[Q]}p(v,b)-\max_{b\in[Q]}p(v,b)\right).
\]
Define the minimum ambiguity over initial labels
\(\mu_L:V\to[0,1]\) by
\[
    \mu_L(s)
    :=
    \min_{a\in[Q]}\operatorname{Amb}(p^a_{L,s}).
\]

The tester looks for scales at which many seeds have large
\(\mu_L(s)\).  In completeness, the minimum over initial labels
lets us start with the label assigned to \(s\) by a nearly satisfying
labeling.  Soundness must instead force large ambiguity for every initial
label on the same set of seeds.

\subsection{Properties of ambiguity and completeness}
\label{subsec:ug-completeness}

The following facts hold for any weighted constraint kernel \(\mathcal H\).
They will therefore apply to the sparsifier once it has been constructed.

\subsubsection{Properties of ambiguity}

We will compare endpoint distributions at different stopping times.  The
following elementary properties control the effect of propagating a
distribution or removing some of its mass.

\begin{lemma}[Elementary properties of ambiguity]
\label{lem:permutation-kernel-ambiguity-monotonicity}
For any finite set \(X\), the functional \(\operatorname{Amb}\) is homogeneous and superadditive on nonnegative
measures on \(X\times[Q]\).  If two such measures \(p,q\) satisfy
\(0\le q\le p\) pointwise, then
\[
 \operatorname{Amb}(p)-\|p-q\|_1\le\operatorname{Amb}(q)\le\operatorname{Amb}(p).
\]
Let \(Y\) be a finite destination set.  Suppose a step from \(x\in X\)
draws a destination \(\rv{y}_x\in Y\) and a permutation \(\rv{\Pi}_x\in S_Q\)
from a joint distribution depending only on \(x\).  It sends \((x,b)\) to
\((\rv{y}_x,\rv{\Pi}_x(b))\).  If \(p'\) is the resulting
measure on \(Y\times[Q]\), then \(\operatorname{Amb}(p')\ge\operatorname{Amb}(p)\).
\end{lemma}

\begin{proof}
At each \(x\), the contribution to ambiguity is
\(\min_{j\in[Q]}\sum_{b\ne j}p(x,b)\).  This gives homogeneity and
superadditivity, and deleting mass reduces each such sum by at most the
mass deleted at \(x\).
For the last assertion, fix \(y\in Y\) and \(c\in[Q]\).  Then
\[
 p'(y,c)=\sum_x\E\!\left[\1[\rv{y}_x=y]\,p(x,\rv{\Pi}_x^{-1}(c))\right]
 \le\sum_x\Pr[\rv{y}_x=y]\max_b p(x,b).
\]
Taking the maximum over \(c\) and summing over \(y\) bounds the total
mass kept by \(\sum_x\max_b p(x,b)\).
Since the step preserves total mass, the mass discarded cannot decrease.
\end{proof}

\subsubsection{Completeness}

A nearly satisfying labeling explains every path that avoids its violated
edges.  Averaging the probability of hitting those edges gives the
completeness bound used by the tester.

\begin{lemma}[Weighted completeness for Unique Games ambiguity]
\label{lem:ug-completeness}
If \(\tau_w(\mathcal H)\le\varepsilon\), then for every \(L\ge1\),
\[
 \E_{\rv{s}\sim\pi}\mu_L(\rv{s})\le\varepsilon L,
 \qquad
 \Prb_{\rv{s}\sim\pi}[\mu_L(\rv{s})\ge\alpha]
       \le\varepsilon L/\alpha\quad(\alpha>0).
\]
\end{lemma}

\begin{proof}
Choose a labeling \(x:V\to[Q]\) with \(2w(B_x)/\vol(V)\le\varepsilon\).
From initial label \(x(s)\), every path to \(v\) that avoids \(B_x\) predicts
\(x(v)\).  Hence \(\mu_L(s)\le\operatorname{Amb}(p^{x(s)}_{L,s})\),
and \(\operatorname{Amb}(p^{x(s)}_{L,s})\) is at most the probability that
the walk from \(s\) traverses an edge of \(B_x\).
A stationary step traverses \(B_x\) with probability
\(2w(B_x)/\vol(V)\le\varepsilon\), so the expected number of such
traversals before \(\rv{T}_L\) is at most \(\varepsilon\E \rv{T}_L=\varepsilon L\).
The tail bound is Markov's inequality.
\end{proof}

\subsection{Trace chains and partial labelings}
\label{subsec:ug-trace-peeling}

In the soundness analysis, a \emph{residual set} \(R\subseteq V\) consists of the vertices
remaining after almost consistently labeled pieces have been removed.
The tester does not perform these removals or simulate the trace chains;
they are used to prove that many seeds have large ambiguity.
The trace chain and the failure probability of a partial labeling, defined
below, let us either remove another piece or prove that walks develop
substantial ambiguity after few returns.

\subsubsection{Trace chain basics}

We first define the trace chain on a residual set and record the reversibility
facts used later for energy and return-time arguments.

Let \(R\subseteq V\) be nonempty.  For a lazy walk started in \(R\), define
the first return time to \(R\) by
\[
    \rv{\tau}_R:=\min\{t\in\mathbb Z_{\ge1}:\rv{v}_t\in R\}.
\]
A trace step records the return vertex \(\rv{v}_{\rv{\tau}_R}\) and applies the
permutation \(\Pi_{\rv{P}_{\rv{\tau}_R}}\) of the return path to the initial label.

We track either the return vertex alone or the return vertex together with
its label.  Write \(K_R\) for the transition kernel on \(R\) and
\(\widehat K_R\) for the transition kernel on \(R\times[Q]\).  For
\(u,v\in R\) and \(b,c\in[Q]\), their entries are
\[
\begin{aligned}
 K_R(u,v)&:=\Pr_u[\rv{v}_{\rv{\tau}_R}=v],\\
 \widehat K_R((u,b),(v,c))&:=\Pr_u[\rv{v}_{\rv{\tau}_R}=v,\ \Pi_{\rv{P}_{\rv{\tau}_R}}(b)=c],
\end{aligned}
\]
where \(\Pr_u\) and \(\E_u\) denote probability and expectation for a walk started at \(u\).

\begin{lemma}[Basic properties of the UG trace chain]
\label{lem:ug-trace-chain-basic}
For every nonempty \(R\subseteq V\), the trace chain is a Markov chain on
\(R\).  If \(C\) is the communicating class of the starting state under
\(M\), then every trace step stays in \(R\cap C\).  The kernel \(K_R\) is reversible with respect to \(\pi_R\).
For all \(u,v\in R\) and \(b,c\in[Q]\),
\[
 \pi_R(u)\widehat K_R((u,b),(v,c))
 =\pi_R(v)\widehat K_R((v,c),(u,b)),\qquad
 \widehat K_R((u,b),(u,b))\ge\tfrac12.
\]
\end{lemma}

\begin{proof}
Starting from \(u\in R\), the walk never leaves its closed communicating
class \(C\).  The finite irreducible walk on \(C\) hits \(R\cap C\) again
with probability one, so the row sums of \(K_R\) are one and different
components do not communicate.

A transition of the trace chain from \(u\) to \(v\) follows a path
\(u=u_0,\ldots,u_j=v\) of the ordinary walk whose internal vertices
avoid \(R\).  Detailed balance for \(M\) gives
\[
 \pi_R(u)\prod_{i=1}^j M(u_{i-1},u_i)
 =\pi_R(v)\prod_{i=1}^j M(u_i,u_{i-1}),
\]
since \(u,v\in R\) and \(\pi_R=\pi/\pi(R)\) on \(R\).
The reversed path applies the inverse of the original path permutation.
Summing over return paths gives detailed balance for \(K_R\).
A path carrying \(b\) to \(c\) carries \(c\) to \(b\) in reverse, so
restricting the sum to these paths gives the identity for \(\widehat K_R\).

With probability at least \(1/2\), the first step stays at \(u\) and
applies the identity permutation.  On this event the trace step returns to \((u,b)\).
\end{proof}

\subsubsection{Failure probability of a partial labeling}

A trace step can fail a partial labeling by leaving its domain or returning
with a different label.  The least such failure probability will control
both the removal of almost consistent pieces and the growth of ambiguity.

For nonempty \(S\subseteq R\) and a partial labeling \(\ell:S\to[Q]\),
start at \(\rv{u}\in S\) drawn from \(\pi_R\) conditioned on \(S\), with label \(\ell(\rv{u})\).
Let \((\rv{v},\rv{c})\in R\times[Q]\) be the vertex and label after one trace step.  Define
\(\beta_R(S,\ell)\in[0,1]\) as the probability that \(\rv{v}\notin S\), or
that \(\rv{v}\in S\) and \(\rv{c}\ne\ell(\rv{v})\).
Define the minimum failure probability \(\beta(R)\in[0,1]\) by
\[
    \beta(R)
    :=
    \min_{\emptyset\neq S\subseteq R,\ \ell:S\to[Q]}
    \beta_R(S,\ell).
\]

\subsection{From trace failures to ambiguity at a geometric stopping time}
\label{subsec:ug-geometric-ambiguity}

We show that if every partial labeling on a nonempty subset of \(R\) fails
a trace step with probability at least \(\theta\), then many seeds have
constant ambiguity for every initial label at a suitable mean walk length.
We track the mass by which a label exceeds half the total probability at
its vertex.  We then compare ambiguity after a fixed number of trace steps
with ambiguity for geometrically stopped walks.

\subsubsection{Bounding ambiguity after trace steps}

Fix a nonempty residual set \(R\subseteq V\).
For a start-label pair \((s,a)\in R\times[Q]\) and an integer \(t\ge0\),
let \(p_t^a:R\times[Q]\to[0,1]\) assign to \((v,b)\) the probability that
\(t\) trace steps end at \(v\) carrying label \(b\).  Its vertex marginal is
\(u_t:R\to[0,1]\), given by \(u_t(v):=\sum_{b\in[Q]}p_t^a(v,b)\).
The ambiguity after these \(t\) steps is \(\operatorname{Amb}(p_t^a)\).
To bound \(\operatorname{Amb}(p_t^a)\), we control the sum over all vertices
and labels of \((p_t^a(v,b)-u_t(v)/2)_+\), the mass above half the vertex
marginal.  Here \(y_+:=\max\{y,0\}\).
The bound will depend on the seed only through
\(\pi_R(s)\), so it will apply to all initial labels on the same set of seeds.

\begin{lemma}[Ambiguity after \(t\) trace steps]
\label{lem:ug-trace-time-majority-excess}
If \(\beta(R)\ge\theta>0\), then for every \(s\in R\),
\(a\in[Q]\), and integer \(t\ge0\),
\begin{equation}
 \operatorname{Amb}(p_t^a)\ge
 \frac{1}{2}\left(1-\frac{e^{-\theta^2t/4}}{\sqrt{\pi_R(s)}}\right).
 \label{eq:ug-trace-ambiguity-decay}
\end{equation}
\end{lemma}

\begin{proof}
Define the reversible measure \(\widehat{\pi}_R:R\times[Q]\to(0,1]\) by
\(\widehat{\pi}_R(v,b):=\pi_R(v)\); its total mass is \(Q\).  For functions
\(f,g:R\times[Q]\to\mathbb R\), define
\[
 \langle f,g\rangle_*
   :=\sum_{v\in R,b\in[Q]}\pi_R(v)f(v,b)g(v,b),\qquad
 \|g\|_*^2:=\langle g,g\rangle_*,\qquad
 \mathcal E(g):=\langle g,(I-\widehat K_R)g\rangle_*.
\]
The detailed balance identity in
Lemma~\ref{lem:ug-trace-chain-basic} makes
\(\widehat K_R\) self-adjoint.  Its identity holding probability is at least
\(1/2\), so \(2\widehat K_R-I\) is a reversible Markov kernel.  Consequently,
\(0\preceq\widehat K_R\preceq I\).

\paragraph{Dirichlet energy of functions supported on at most one label per vertex.}
We first establish a Cheeger bound for nonnegative functions \(g:R\times[Q]\to\mathbb R_{\ge0}\) whose
support contains at most one label at each vertex.  Every nonempty positive
superlevel set of \(g^2\) is the graph of a partial labeling
\(\ell:S\to[Q]\) on a nonempty \(S\subseteq R\).
The graph of \(\ell\) has \(\widehat\pi_R\)-measure \(\pi_R(S)\).
Conditional on starting there under this measure, a trace step
leaves the graph with probability \(\beta_R(S,\ell)\).  Thus the flow
leaving this superlevel set is
\(\pi_R(S)\beta_R(S,\ell)\ge\theta\pi_R(S)\).
For \(x,y\in R\times[Q]\), put
\(c_{xy}:=\widehat{\pi}_R(x)\widehat K_R(x,y)=c_{yx}\ge0\).
Integrating over these superlevel sets and then applying Cauchy--Schwarz gives
\[
 \theta\|g\|_*^2
 \le\frac{1}{2}\sum_{x,y}c_{xy}|g(x)^2-g(y)^2|
 \le\sqrt{2\mathcal E(g)\|g\|_*^2}.
\]
Here \(\mathcal E(g)=\tfrac12\sum_{x,y}c_{xy}(g(x)-g(y))^2\), and
\(\tfrac12\sum_{x,y}c_{xy}(g(x)+g(y))^2\le2\|g\|_*^2\).
Thus, including the case \(g=0\),
\begin{equation}
 \mathcal E(g)\ge\frac{\theta^2}{2}\|g\|_*^2.
 \label{eq:ug-restricted-cheeger}
\end{equation}
This bound only concerns functions with the stated support restriction.

\paragraph{Contraction of the squared norm.}
We divide by the stationary probabilities \(\pi_R(v)\) so that reversibility lets the kernels
act directly on the densities \(f_t:R\times[Q]\to\mathbb R_{\ge0}\)
and \(w_t:R\to\mathbb R_{\ge0}\), defined by
\(f_t(v,b):=p_t^a(v,b)/\pi_R(v)\) and \(w_t(v):=u_t(v)/\pi_R(v)\):
\[
 f_{t+1}=\widehat K_Rf_t,\qquad w_{t+1}=K_Rw_t.
\]
To obtain the support restriction needed for
\eqref{eq:ug-restricted-cheeger}, keep only the positive excess over half
the mass at each vertex.  Define \(h_t:R\times[Q]\to\mathbb R\) and
\(g_t:R\times[Q]\to\mathbb R_{\ge0}\) by
\[
 h_t(v,b):=2f_t(v,b)-w_t(v),\qquad g_t(v,b):=(h_t(v,b))_+,
\]
so \(g_t(v,b)=2(p_t^a(v,b)-u_t(v)/2)_+/\pi_R(v)\).  A positive value requires
\(p_t^a(v,b)>u_t(v)/2\), which holds for at most one label at each vertex.
A function on \(R\) is identified here with the function constant across
labels at each vertex; under this identification
\(\widehat K_R w_t=K_Rw_t\).  Hence \(h_{t+1}=\widehat K_R h_t\).
Since the kernel is nonnegative,
\(g_{t+1}\le\widehat K_R g_t\) pointwise.  Therefore
\begin{equation}
 \begin{aligned}
 \|g_{t+1}\|_*^2
 &\le\|\widehat K_R g_t\|_*^2
 \le\langle g_t,\widehat K_R g_t\rangle_*\\
 &\le(1-\theta^2/2)\|g_t\|_*^2.
 \end{aligned}
 \label{eq:ug-majority-excess-contraction}
\end{equation}
The second inequality uses \(0\preceq\widehat K_R\preceq I\), and the last
uses \eqref{eq:ug-restricted-cheeger}.  Initially \(g_0\) equals
\(1/\pi_R(s)\) at \((s,a)\) and is zero elsewhere, so
\[
 \|g_t\|_*^2\le e^{-\theta^2t/2}/\pi_R(s).
\]

\paragraph{Bounding \(\operatorname{Amb}(p_t^a)\) using \(g_t\).}
The support of \(g_t\) has \(\widehat{\pi}_R\)-measure at most \(1\).
Cauchy--Schwarz and the bound on \(\|g_t\|_*\) imply
\[
 1-2\operatorname{Amb}(p_t^a)
 \le\sum_v(2\max_b p_t^a(v,b)-u_t(v))_+
 =\sum_{v,b}\pi_R(v)g_t(v,b)
 \le\|g_t\|_*
 \le\frac{e^{-\theta^2t/4}}{\sqrt{\pi_R(s)}}.
\]
Rearranging proves \eqref{eq:ug-trace-ambiguity-decay}.
\end{proof}

\begin{remark}
The label lift need not contract every function that sums to zero across
labels at each vertex.  For \(Q=4\), put the permutation \((12)(34)\) on
each edge of a triangle.
No labeling satisfies the instance, and every partial labeling has positive
trace failure probability for \(R=V\).  Nevertheless, the function taking
values \((1,1,-1,-1)\) at every vertex is invariant under the label lift
and sums to zero across labels at each vertex.  The proof of
Lemma~\ref{lem:ug-trace-time-majority-excess} establishes contraction for \(g_t\), whose support contains at most one label per vertex,
as needed for \eqref{eq:ug-restricted-cheeger}.
\end{remark}

\subsubsection{A direct comparison with geometric time}

For each integer \(k\ge1\), let \(\rv{H}_k\in\mathbb Z_{\ge1}\) be the number of ordinary walk steps
up to and including the \(k\)-th return to \(R\).
For a walk from \((s,a)\), we compare the ambiguity of the distribution
\(p_k^a\) after the \(k\)-th return with that of \(p^a_{L,s}\) at geometric
stopping.  Discarding the paths that stop before the \(k\)-th return
reduces ambiguity by at most \(\Prb_s[\rv{T}_L<\rv{H}_k]\).

\begin{lemma}[Trace-to-geometric comparison]
\label{lem:trace-geometric-comparison}
For every \(s\in R\), \(a\in[Q]\), integer \(k\ge1\), and \(L\ge1\),
\begin{equation}
 \operatorname{Amb}(p^a_{L,s})
 \ge \operatorname{Amb}(p_k^a)-\frac{\E_s\rv{H}_k}{L+1}.
 \label{eq:trace-geometric-loss}
\end{equation}
Here \(p_k^a\) is the distribution after \(k\) trace steps from \((s,a)\).
In particular, \(\mu_L(s)\ge\min_{a\in[Q]}\operatorname{Amb}(p_k^a)
-\E_s\rv{H}_k/(L+1)\).
\end{lemma}

\begin{proof}
The trace endpoint distribution satisfies
\(p_k^a(v,b)=\Prb_s[\rv{v}_{\rv{H}_k}=v,\Pi_{\rv{P}_{\rv{H}_k}}(a)=b]\).
Define the measure \(q:R\times[Q]\to[0,1]\) by keeping only paths with
\(\rv{T}_L\ge \rv{H}_k\):
\[
 q(v,b):=\Prb_s[\rv{v}_{\rv{H}_k}=v,\Pi_{\rv{P}_{\rv{H}_k}}(a)=b,\rv{T}_L\ge \rv{H}_k].
\]
Let \(G_L\) be the endpoint kernel on \(V\times[Q]\) for
the geometrically stopped walk, and let \(r:V\times[Q]\to[0,1]\) be the measure
of endpoints and labels of paths stopped before \(\rv{H}_k\).  Extend measures
on \(R\times[Q]\) by zero outside \(R\).  Geometric memorylessness and
the strong Markov property give, for every \(v\in V\) and \(b\in[Q]\),
\[
 p^a_{L,s}(v,b)=(qG_L)(v,b)+r(v,b).
\]
Since \(0\le q\le p_k^a\), the properties in
Lemma~\ref{lem:permutation-kernel-ambiguity-monotonicity} imply
\[
 \begin{aligned}
 \operatorname{Amb}(p^a_{L,s})
 &\ge\operatorname{Amb}(qG_L)\ge\operatorname{Amb}(q)\\
 &\ge\operatorname{Amb}(p_k^a)-\Prb_s[\rv{T}_L<\rv{H}_k].
 \end{aligned}
\]
Since \(\rv{T}_L\) is independent of the walk,
\[
 \Prb_s[\rv{T}_L<\rv{H}_k]
 =\E_s\left[1-\left(\frac{L}{L+1}\right)^{\rv{H}_k}\right]
 \le\frac{\E_s\rv{H}_k}{L+1}.
\]
The inequality uses \(1-(L/(L+1))^h\le h/(L+1)\) for integers \(h\ge0\).
The early-stopping probability \(\Prb_s[\rv{T}_L<\rv{H}_k]\) depends only
on the vertex walk and its stopping time, so its upper bound is uniform in \(a\).
\end{proof}

\begin{lemma}[Constant ambiguity at a geometric stopping time]
\label{lem:ug-trace-to-geometric}
There is an absolute constant \(C>0\) such that the following holds.
If \(\beta(R)\ge\theta>0\), then for every
\[
 L\ge C\frac{\theta^{-2}\log(2n)}{\pi(R)},
\]
a set of starts of \(\pi_R\)-mass at least \(1/2\) satisfies
\(\mu_L(s)\ge1/4\).
\end{lemma}

\begin{proof}
Put
\(k:=\lceil2\theta^{-2}\log(64n)\rceil\).
For each \(s\in R\), Lemma~\ref{lem:ug-trace-time-majority-excess}
bounds the shortfall of \(\operatorname{Amb}(p_k^a)\) from \(1/2\) by
\(e^{-\theta^2 k/4}/(2\sqrt{\pi_R(s)})\) for every initial label \(a\).
Lemma~\ref{lem:trace-geometric-comparison} adds at most
\(\E_s\rv{H}_k/(L+1)\) for early stopping.  Both terms are uniform in
\(a\), so taking the minimum over \(a\) gives
\[
 \mu_L(s)\ge\frac{1}{2}-D(s),\qquad
 D(s):=\frac{e^{-\theta^2k/4}}{2\sqrt{\pi_R(s)}}
                         +\frac{\E_s\rv{H}_k}{L+1}.
\]
The function \(D\) has small mean under \(\pi_R\): Cauchy--Schwarz gives
\(\sum_{s\in R}\sqrt{\pi_R(s)}\le\sqrt{|R|}\le\sqrt n\), and
Lemma~\ref{lem:stationary-return-time} bounds the mean return time.  Hence
\[
 \E_{\rv{s}\sim\pi_R}D(\rv{s})
 \le\frac{\sqrt n}{2}e^{-\theta^2k/4}+\frac{k}{\pi(R)(L+1)}
 \le\frac{1}{8}
\]
when \(L\ge16k/\pi(R)\).  Markov's inequality now gives
\(\Prb_{\rv{s}\sim\pi_R}[\mu_L(\rv{s})<1/4]\le4\E D(\rv{s})\le1/2\).
Since \(\theta\le1\), an absolute choice of \(C\) makes the stated lower
bound on \(L\) sufficient.
\end{proof}

\subsubsection{Finding a residual and choosing the scales}

We now find a residual set to which
Lemma~\ref{lem:ug-trace-to-geometric} applies.  Repeatedly remove the domain
of a partial labeling with small failure probability, allowing the threshold
to grow as the remaining \(\pi\)-mass decreases.  If the remaining mass
becomes small, combining the partial labelings gives a labeling with few
violations, contradicting the distance assumption.  The following lemma
chooses the thresholds so that on the remaining \(R\), every partial
labeling fails with probability at least \(\lambda/(2\pi(R))\), for the
parameter \(\lambda\) below.

\begin{lemma}[A residual set with large minimum failure probability]
\label{lem:ug-multiscale-trace-residual}
There is an absolute constant \(c_\lambda\in(0,1)\) such that the following
holds.  Let \(0<\rho<1/10\), and set
\[
 \lambda:=c_\lambda\frac{\rho}{1+\log(1/\rho)}.
\]
If \(\tau_w(\mathcal H)\ge\rho\), there is a nonempty
\(R\subseteq V\) such that
\[
 \pi(R)\ge\lambda,\qquad \beta(R)\ge\lambda/(2\pi(R)).
\]
\end{lemma}

\begin{proof}
For any nonempty residual set \(R\subseteq V\) and partial labeling
\(\ell:S\to[Q]\) with \(\emptyset\ne S\subseteq R\),
let \(b_{\rm int}\) be the total weight of edges inside \(S\) violated by
\(\ell\).  Consider failures in which the ordinary walk returns to \(R\)
in one step.  Each violated edge inside \(S\) gives two such directed
steps, one in each direction.  Each edge from \(S\) to \(R\setminus S\)
gives one, directed out of \(S\).  Since
\(\pi_R(u)M(u,v)=w_{uv}/\vol(R)\),
\[
 2b_{\rm int}+w(S,R\setminus S)\le\beta_R(S,\ell)\vol(S).
\]
Start with \(R_0:=V\).  At iteration \(i\in\mathbb Z_{\ge0}\), put \(z_i:=\pi(R_i)\).
While \(z_i\ge\lambda\), stop if
\(\beta(R_i)\ge\lambda/(2z_i)\).  Otherwise choose a partial
labeling \(\ell_i:S_i\to[Q]\) on a nonempty \(S_i\subseteq R_i\) with
\(\beta_{R_i}(S_i,\ell_i)<\lambda/(2z_i)\), and
set \(R_{i+1}:=R_i\setminus S_i\).

Suppose the process reaches a residual \(R_*\) of mass less than
\(\lambda\).  Label each removed set \(S_i\) by \(\ell_i\), and label \(R_*\)
arbitrarily.  Charge edges violated inside \(S_i\) to \(S_i\), and charge
each edge from \(S_i\) to \(R_i\setminus S_i\) to \(S_i\) when it is removed.
The inequality for \(2b_{\rm int}+w(S_i,R_i\setminus S_i)\) bounds the
charge to \(S_i\) by \(\lambda\vol(S_i)/(2z_i)\); edges inside \(R_*\)
have total weight at most \(\vol(R_*)/2\).  Normalizing by \(\vol(V)/2\)
gives
\[
 \tau_w(\mathcal H)
 \le\lambda\sum_i\frac{z_i-z_{i+1}}{z_i}+\pi(R_*)
 \le\lambda\bigl(2+\log(1/\lambda)\bigr)<\rho
\]
for sufficiently small absolute \(c_\lambda\), a contradiction.
For the second inequality, every term except the final one is at most
\(\int_{z_{i+1}}^{z_i}dz/z\), with these intervals in \([\lambda,1]\);
the final term is at most one.  Each removal is nonempty, so the process
must instead stop at the claimed residual.
\end{proof}

The algorithm does not know the residual mass \(\pi(R)\), so it tries dyadic
scales.  For \(L\) at least a sufficiently large constant times
\(\pi(R)\lambda^{-2}\log(2n)\), Lemma~\ref{lem:ug-trace-to-geometric}
gives \(\mu_L(s)\ge1/4\) on a set of seeds of \(\pi\)-mass at least
\(\pi(R)/2\).  At a scale \(r\) with \(\pi(R)\le r\le2\pi(R)\), we can
use \(r\) in place of \(\pi(R)\) to choose \(L\), and this set has mass
at least \(r/4\).

\begin{lemma}[A scale with many seeds of constant ambiguity]
\label{lem:ug-multiscale-ambiguity}
Let \(0<\rho<1/10\), and let \(\lambda\) be as in
Lemma~\ref{lem:ug-multiscale-trace-residual}.  Define the finite set of scales
\(\mathcal Q_\rho\subseteq(0,1]\) and, for each scale \(r\), the
expected number \(L_r\ge1\) of steps in the geometrically stopped walk by
\[
 \mathcal Q_\rho:=\{2^j\lambda:j\in\mathbb Z_{\ge0},\ 2^j\lambda\le1\}
                    \cup\{1\},\qquad
 L_r:=C r\lambda^{-2}\log(2n)\quad(r\in\mathcal Q_\rho),
\]
where \(C\) is a sufficiently large absolute constant.
If \(\tau_w(\mathcal H)\ge\rho\), then some
\(r\in\mathcal Q_\rho\) satisfies
\[
 \Prb_{\rv{s}\sim\pi}[\mu_{L_r}(\rv{s})\ge1/4]\ge r/4.
\]
\end{lemma}

\begin{proof}
Lemma~\ref{lem:ug-multiscale-trace-residual} supplies a residual \(R\) of
mass \(\pi(R)\ge\lambda\) with \(\beta(R)\) at least
\(\theta:=\lambda/(2\pi(R))\).  Choose \(r\in\mathcal Q_\rho\) with
\(\pi(R)\le r\le2\pi(R)\).  Since \(\theta^{-2}/\pi(R)=4\pi(R)/\lambda^2\), choosing
\(C\) sufficiently large makes \(L_r\) satisfy the lower bound on \(L\)
of Lemma~\ref{lem:ug-trace-to-geometric}.
That lemma gives \(\mu_{L_r}(s)\ge1/4\) on a set of seeds of \(\pi_R\)-mass at least
\(1/2\), hence on global \(\pi\)-mass at least \(\pi(R)/2\ge r/4\).
\end{proof}

\subsection{Sparsification with the original stationary distribution}
\label{sparse:sec:sparsification}

To reduce the cost of estimating ambiguity, we use independent edge sampling
as in the Label Cover sparsification argument of Dinitz, Kortsarz, and
Raz~\cite[Lemma~2.1]{DinitzKortsarzRaz}.  Our walk on the sampled graph uses
identity holding transitions to preserve the input stationary distribution.

Let \(G=(V,E)\) be the positive-degree part of the input \(\mathcal U\), and
let \(n\) and \(m\) denote the numbers of vertices and edges in the original input.
Assume \(m>0\) and
\(0<\varepsilon<\rho<1\), and write \(\pi(v):=\deg_G(v)/(2m)\).  Define the integer \(\ell_Q\) and the
edge sampling probability \(p\in(0,1]\) by
\begin{equation}
 \ell_Q:=\lceil\log_2(2nQ)\rceil,\qquad
 p:=\min\left\{1,\frac{6000n\ell_Q}{\rho m}\right\}.
 \label{sparse:eq:sparsification-parameters}
\end{equation}
Call a vertex \(v\in V\) \emph{eligible} if \(p=1\) or \(p\deg_G(v)\ge60\ell_Q\).
An edge is eligible if both endpoints are.  We discard the ineligible
edges, of which there are at most \(\rho m/100\).  The degree threshold ensures that, with high
probability, every eligible vertex \(v\) retains at most \(2p\deg_G(v)\)
edges.  For each eligible edge \(e\), draw
an independent retention indicator \(\rv{Z}_e\in\{0,1\}\) with
\(\Prb[\rv{Z}_e=1]=p\), and set
\(\rv{Z}_e=0\) on ineligible edges.  Write \(\rv{Z}_{uv}:=\rv{Z}_{\{u,v\}}\)
on edges and \(\rv{Z}_{uv}:=0\) on nonedges.  The sparsifier \(\rv{H}\)
assigns weight \(\rv{w}_e:=\rv{Z}_e/p\) to each edge and keeps its original
permutation.  All \(Q\) lifted transitions of \(e\) use the same coin \(\rv{Z}_e\).
The expected number of retained edges is at most \(pm=O(n\log(nQ)/\rho)\).
For a realization \(H\) of the sparsifier, let \(k_H(v)\) count the retained
edges at \(v\).  To keep the walk lazy, we will use realizations in the set
\begin{equation}
 \mathcal G:=\{H:k_H(v)\le2p\deg_G(v)\text{ for every }v\in V\}.
 \label{sparse:eq:degree-event}
\end{equation}

\begin{lemma}[preserving the promise]
\label{sparse:lem:sparsification-promise}
We have \(\Prb[\rv{H}\notin\mathcal G]\le1/100\).  If \(\tau(\mathcal U)\ge\rho\), then,
with probability at least \(99/100\), every labeling violates edges of
total weight at least
\(\rho m/4\) in \(\rv{H}\).  If \(\tau(\mathcal U)\le\varepsilon\), then, with
probability at least \(99/100\), some labeling violates edges of total
weight at most
\(100\varepsilon m\) in \(\rv{H}\).
\end{lemma}

\begin{proof}
For \(p=1\) all edges are retained.  Suppose \(p<1\).  At an eligible
vertex \(v\), the retained degree \(k_{\rv{H}}(v)\) is dominated by
\(\operatorname{Bin}(\deg_G(v),p)\).  At an ineligible vertex \(v\), \(k_{\rv{H}}(v)=0\).
Chernoff's inequality and a union bound give
\(\Prb[\rv{H}\notin\mathcal G]\le n e^{-20\ell_Q}<1/100\).

Ineligible vertices have total degree at most \(60n\ell_Q/p=\rho m/100\).
In a no-instance each labeling therefore violates at least \(.99\rho m\)
eligible edges.  A Chernoff bound for each labeling and a union bound over
at most \(Q^n\) labelings show that the probability that some labeling violates
fewer than \(p\rho m/4\) retained edges is at most \(Q^n\exp(-.99p\rho m/8)<1/100\).
In a yes-instance, a labeling violating at most \(\varepsilon m\) input
edges has expected violated weight at most \(\varepsilon m\) after sampling.
Markov's inequality gives the claim.
\end{proof}

For a realization \(H\), write \(Z_e\) and \(w_e\) for its edge indicators and weights.  For \(H\in\mathcal G\), define the weighted constraint kernel
\(\mathcal H_H:=(V,w,D,\{\pi_{uv}\})\) by \(D_v:=4\deg_G(v)\).
Its base kernel \(P_H:V\times V\to[0,1]\) is defined by
\begin{equation}
 P_H(u,v):=\frac{Z_{uv}}{4p\deg_G(u)}\quad(u\ne v),\qquad
 P_H(u,u):=1-\frac{k_H(u)}{4p\deg_G(u)}.
 \label{sparse:eq:sparse-kernel}
\end{equation}
The kernel \(P_H\) has identity holding probability at least \(1/2\) and
stationary distribution \(\pi\).  Ineligible vertices are absorbing.
The volume of \(\mathcal H_H\) is \(\sum_v D_v=8m\), so
\(\tau_w(\mathcal H_H)=\min_x w(B_x)/(4m)\), where \(w(B_x)\) is the
total weight of edges violated by \(x\).  Lemma~\ref{sparse:lem:sparsification-promise}
gives \(\Prb[\rv{H}\notin\mathcal G]\le1/100\).
In the no case, \(\min_x w(B_x)\ge\rho m/4\) with probability at least
\(99/100\), which gives \(\tau_w(\mathcal H_H)\ge\rho/16\) whenever
\(H\in\mathcal G\) and this lower bound holds.  In the yes case,
\(\min_x w(B_x)\le100\varepsilon m\) with probability at least \(99/100\),
which similarly gives \(\tau_w(\mathcal H_H)\le25\varepsilon\) whenever
\(H\in\mathcal G\) and this upper bound holds.  We write
\(\mu_L^H:V\to[0,1]\) for the minimum ambiguity over initial labels
for geometrically stopped walks on \(\mathcal H_H\).

\subsection{Local simulation of the random sparsifier}
\label{sparse:sec:deferred-oracle}

The algorithm must use the same retained graph in its endpoint walks, reverse
walks, and forward pushes.  Independent resampling at successive queries would
not provide this consistency.  We implement a local oracle that maintains the
conditional distribution of the unknown edge coins \(\rv{Z}_e\).
When \(p<1\), the coin of an eligible edge remains unknown until a walk
crosses it or the oracle reveals the neighborhood of one of its endpoints.
In continuous time the remaining unknown coins are independent given the
observed walks and neighborhoods.

For any realization \(H\), let
\(\mathcal L_H:V\times V\to\mathbb R\) be the continuous-time generator
defined by
\[
 \mathcal L_H(u,v):=\frac{Z_{uv}}{p\deg_G(u)}\quad(u\ne v),\qquad
 \mathcal L_H(u,u):=-\frac{k_H(u)}{p\deg_G(u)}.
\]
A jump from \(u\) to \(v\) applies \(\pi_{uv}\) to the current label;
holding preserves it.  For \(H\in\mathcal G\),
\(P_H=I+\mathcal L_H/4\).  A continuous-time walk stopped at an independent
exponential time of rate \(4/L\) has the same endpoint-label distribution
as a \(P_H\)-walk stopped after \(\rv{T}_L\) steps, since on the label lift
\begin{equation}
 \frac{1}{L+1}\left(I-\frac{L}{L+1}P_H\right)^{-1}
 =\left(I-\frac{L}{4}\mathcal L_H\right)^{-1}.
 \label{sparse:eq:continuous-resolvent}
\end{equation}
Here and below this identity uses the lifted versions of the displayed base
operators when the transported label is recorded.

We describe the oracle state for a fixed history of observations.
The two operations, \textsc{Walk} (Algorithm~\ref{alg:walk}) and
\textsc{Expose} (Algorithm~\ref{alg:expose}), share a cache of revealed
retained edges with their permutations in both directions and a set
\(F\subseteq V\) of vertices whose neighborhoods have been fully exposed.
Both are initially empty.  When \(p<1\), the coin of an eligible edge is
known if the edge has been traversed or has an endpoint in \(F\).  Thus a missing
cache entry for an edge incident to \(F\) denotes a zero coin.  Since the
input graph is simple, the unordered pair of endpoints identifies each edge.

When \(p<1\), the oracle also keeps track of the total time \(t_v\ge0\)
that the walks have spent at each vertex \(v\), including the elapsed part
of the current visit.  Initially every \(t_v\) is zero.
An input edge \(e=uv\), if retained, has jump rate \(1/(p\deg_G(u))\)
out of \(u\) and \(1/(p\deg_G(v))\) out of \(v\).  Its integrated jump
rate \(\alpha_e\ge0\) adds up the two rates weighted by the observed times,
\[
 \alpha_e:=\frac{t_u}{p\deg_G(u)}+\frac{t_v}{p\deg_G(v)}.
\]
For an eligible edge whose coin is still unknown, no walk has crossed
\(e\).  Think of a clock for \(e\) that runs only while the walk is at
one of its endpoints.  If \(e\) is absent, its clock never rings.  If
\(e\) is present, its clock must have stayed silent during every visit
observed so far.  The exponential distribution of the waiting time gives a factor
\(e^{-\alpha_e}\) for these silent intervals.  Thus the two possibilities
have relative weights \(1-p\) for absence and \(p e^{-\alpha_e}\) for
presence.  Normalizing gives the conditional retention probability
\begin{equation}
 q_e:=\frac{p e^{-\alpha_e}}{1-p+p e^{-\alpha_e}}\in(0,p].
 \label{sparse:eq:posterior-coins}
\end{equation}
Initially \(q_e=p\), and \(q_e\) decreases as the walks spend time at either
endpoint without crossing \(e\).  This explains the factor \(q_e/p\) in
both algorithms.  In \textsc{Walk}, a rate-one clock proposes a uniform
neighbor of \(u\), so it proposes \(e\) at rate \(1/\deg_G(u)\).
Accepting with probability \(q_e/p\) gives
the conditional jump rate \(q_e/(p\deg_G(u))\).  For \textsc{Expose},
selecting an index with probability \(p\) and then accepting with probability
\(q_e/p\) retains \(e\) with probability \(q_e\).

\begin{procedure}{Walk\((u,L)\)}{alg:walk}
\procedureio{A starting vertex \(u\in V\) and a mean discrete walk length \(L\ge1\).}
  {A random pair \((\rv{v},\rv{\Gamma})\in V\times S_Q\) consisting of the endpoint
  and the permutation transporting labels from the start to that endpoint.}
\begin{enumerate}[leftmargin=*,nosep]
\item Use \(u\) for the current vertex and \(\Gamma\in S_Q\) for the
transport along the path observed so far.  Set \(\Gamma\) to the identity
and draw an independent exponential
stopping time of rate \(4/L\).  On every jump \(u\to v\) below, set
\(\Gamma\gets\pi_{uv}\circ\Gamma\) and \(u\gets v\).
\item If \(p=1\), jump to a uniform original neighbor at rate one until
the stopping time, then return \((u,\Gamma)\).
\item For \(p<1\), repeat until the stopping time.  At an ineligible \(u\),
wait until the stop.  At an eligible \(u\), run two independent clocks.
The first rings at rate \(k_+(u)/(p\deg_G(u))\), where \(k_+(u)\) counts
the cached retained edges incident to \(u\).  The second rings at rate one.
Add the elapsed time to \(t_u\) before processing the next event,
including a stop.
\begin{itemize}[leftmargin=*,nosep]
\item When the first clock rings, jump along a uniformly chosen cached retained edge.
\item When the second clock rings, query a uniform original neighbor.
Write \(v\) for the returned vertex.  If \(e=uv\) is eligible and its coin is unknown, accept with probability
\(q_e/p\) from~\eqref{sparse:eq:posterior-coins}.  On acceptance, cache
\(e\) with \(\pi_{uv}\) and \(\pi_{vu}=\pi_{uv}^{-1}\) and jump to \(v\).
Otherwise stay at \(u\).  A null proposal leaves the cache unchanged.
\end{itemize}
\item At the stopping time, return \((u,\Gamma)\).
\end{enumerate}
\end{procedure}

\begin{procedure}{Expose\((u)\)}{alg:expose}
\procedureio{A vertex \(u\in V\).}
  {All retained neighbors of \(u\) and their oriented constraints.}
\begin{enumerate}[leftmargin=*,nosep]
\item If \(u\in F\), return its cached neighborhood.
\item If \(p=1\), query and cache all original neighbors of \(u\) with
their permutations.  Otherwise, select original adjacency-list indices
independently with probability \(p\) by skipping geometrically distributed runs of unselected indices.  Query each
selected index.  For each resulting eligible edge \(e=uv\) whose coin is
unknown, independently retain it with probability \(q_e/p\) and cache it
with its permutations if retained.
\item Add \(u\) to \(F\) and return all its cached retained neighbors and
constraints, including those known before this call.
\end{enumerate}
\end{procedure}

\begin{lemma}[local oracle for a shared sparsifier]
\label{sparse:lem:deferred-oracle}
For any adaptive sequence of calls, the joint distribution of the replies to
\textup{\textsc{Walk}} and \textup{\textsc{Expose}} is that obtained by sampling \(\rv{H}\)
once in advance, revealing the requested neighborhoods, and using fresh
walk randomness conditional on \(\rv{H}\).
For every realization \(H\in\mathcal G\) and \(L\ge1\), conditional on
the preceding history and \(\rv{H}=H\), a walk costs \(O(L)\) expected
degree and constrained neighbor queries from any start, and the first
exposure of \(u\) costs \(O(1+p\deg_G(u))\) expected queries.
Later enumerations of that neighborhood use a cache of size \(k_H(u)\).
\end{lemma}

\begin{proof}
When \(p=1\), the walk has jump rate one and mean duration \(L/4\), and
the first exposure of \(u\) queries its degree and its neighbors.  This
gives both the stated distribution of the replies and the query bounds.  Suppose \(p<1\).

Imagine drawing the whole graph \(\rv{H}\) before the first call.  We
first ask what the simulated walks and the exposed neighborhoods tell
the oracle about that hidden graph.  We then check that the two oracle
operations generate the next observation with the correct distribution.

\paragraph{What does a walk reveal about an edge?}
Fix the visited vertices and the observed jump and stopping times.  For
each candidate graph \(H\), we ask how likely a walk on \(H\) is to
produce this same record.  Since the times are continuous, we compare
densities at the recorded times.  This density, viewed as a function of
the candidate graph, is the \emph{likelihood}.

Consider one visit to \(u\) lasting time \(t>0\) and ending with a jump
along \(f\).  Set aside the independent stopping clock for the moment.
On \(H\), give each retained edge at \(u\) an independent exponential
clock of rate \(1/(p\deg_G(u))\).  To see the recorded jump,
all the clocks must first stay silent for \(t\) units of time.  The clock
for \(f\) must then ring.  Writing \(Z_e\in\{0,1\}\) for the indicator
of \(e\) in \(H\), the density for this jump is
\[
 \underbrace{\prod_{e\ni u}\exp\!\left(-\frac{Z_e t}{p\deg_G(u)}\right)}
             _{\text{probability of no jump for time }t}
 \ \underbrace{\frac{Z_f}{p\deg_G(u)}}_{\text{jump rate along }f}.
\]
Multiplying this density by a small interval width \(\Delta>0\) gives,
to first order, the probability that the next jump uses \(f\) and occurs
between times \(t\) and \(t+\Delta\) after arrival at \(u\).
For example, compare two graphs that both contain \(f\) and differ only
on another edge \(e\) at \(u\).  The graph containing \(e\) assigns the
same visit a density smaller by a factor \(e^{-t/(p\deg_G(u))}\), because
one extra clock had to stay silent.  The stopping clock contributes the
same factors for both graphs.  If the visit ends at the stopping time,
the edge clocks contribute only the probability of no jump during the
visit, \(\prod_{e\ni u}\exp(-Z_e t/(p\deg_G(u)))\).

Let \(h\) record the visited vertices, the jump times and elapsed waiting
times up to the current time, the completed walks' stopping times, and
the neighborhoods exposed so far.  A caller of \(\textsc{Walk}\)
receives only the endpoint and transport permutation.  The oracle uses
the internal walk record in \(h\) to compute \(t_u\) at each vertex \(u\).
Consider an eligible edge \(e=uv\) that has never been
crossed or exposed.  If \(e\) is present, its clock must have stayed
silent during every visit to either endpoint.  Since \(t_u\) and \(t_v\)
are the total times spent at the two endpoints, the waiting factors for
\(e\) multiply to
\[
 \exp\!\left(-\frac{t_u}{p\deg_G(u)}\right)
 \exp\!\left(-\frac{t_v}{p\deg_G(v)}\right)=e^{-\alpha_e}.
\]
If \(e\) is absent, its clock is silent with probability one.  In the
last column, \(B_e(z)\) is the prior probability of state \(z\) times
the factor that \(e\) contributes to the likelihood of \(h\) in that state.
\[
 \begin{array}{c|c|c|c}
  \text{state }z & \Prb[\rv{Z}_e=z] & \text{likelihood factor}
       & \text{product }B_e(z)\\ \hline
  0\ \text{(absent)}  & 1-p & 1              & 1-p\\
  1\ \text{(present)} & p   & e^{-\alpha_e} & p e^{-\alpha_e}
 \end{array}
\]
Compare two graphs compatible with \(h\) that differ only at \(e\).
The other edges contribute the same factors to the prior and to the
likelihood in both graphs.  For every compatible choice of the other
coins, Bayes' rule therefore gives \(B_e(1)/(B_e(0)+B_e(1))\) for
the probability that \(e\) is present.  Averaging over the other coins
given \(h\) leaves this probability unchanged,
\[
 \Prb[\rv{Z}_e=1\mid h]
 =\frac{B_e(1)}{B_e(0)+B_e(1)}
 =\frac{p e^{-\alpha_e}}{1-p+p e^{-\alpha_e}}=q_e.
\]

Now let \(U\) be the set of all eligible edges whose coins are still
unknown.  Every crossed edge is present, and every exposed neighborhood
fixes its incident coins.  The next call uses only the observed history,
and the path determines the transport permutation.  Given \(h\), the
conditional probability of an assignment \(z\in\{0,1\}^U\) is therefore
proportional to \(\prod_{e\in U}B_e(z_e)\).  Summing over both values
of every unknown coin gives the denominator below.
\[
 \Prb\bigl[(\rv{Z}_e)_{e\in U}=z\mid h\bigr]
 =\frac{\prod_{e\in U}B_e(z_e)}{\prod_{e\in U}(B_e(0)+B_e(1))}
 =\prod_{e\in U}\frac{B_e(z_e)}{B_e(0)+B_e(1)}.
\]
The joint probability is the product of the individual probabilities,
so the unknown coins remain independent given \(h\).

\paragraph{Why the two oracle operations agree with the pre-sampled graph.}
A real jump in the pre-sampled graph changes the walk's position.
Algorithm~\ref{alg:walk} also makes proposals that leave it in place.
To compare the two experiments, keep all the real jumps and add these
null proposals to the record.  At an eligible vertex \(u\), the rate-one
proposal clock in the algorithm should propose each original neighbor at
rate \(1/\deg_G(u)\).  For an unknown eligible edge \(e\) incident to \(u\),
the real jump rate is \(1/(p\deg_G(u))\) if \(e\) is present and zero
if \(e\) is absent.  Averaging these rates gives
\(q_e\cdot 1/(p\deg_G(u))+(1-q_e)\cdot 0=q_e/(p\deg_G(u))\).  To bring the
total proposal rate along \(e\) to \(1/\deg_G(u)\), add null proposals
along \(e\) at rate
\[
 r_e:=\frac{1}{\deg_G(u)}-\frac{q_e}{p\deg_G(u)}
 =\frac{1-q_e/p}{\deg_G(u)}.
\]
For every other original neighbor, add null proposals at rate
\(1/\deg_G(u)\).  In Algorithm~\ref{alg:walk}, jumps along cached
retained edges use the first clock.

What does the oracle learn by also recording the null proposals?
Consider a wait of length \(t\) at \(u\).  If \(R(s)\) is the sum of the
added null-proposal rates after \(s\) units of this wait, observing no
null proposal during the wait contributes the likelihood factor
\(\exp(-\int_0^t R(s)\,ds)\).  This factor is the same for every graph
compatible with the record because the added rates depend only on the
elapsed times and the cache.  It therefore cancels in Bayes' rule.
For an unknown edge \(e\) at \(u\), the same wait adds
\(t/(p\deg_G(u))\) to \(\alpha_e\), so \(q_e\)
in~\eqref{sparse:eq:posterior-coins} already accounts for the time
spent at \(u\) without a real jump.

In the experiment that also records null proposals, write \(h^-\) for
the internal history just before the event at the end of this wait.
It includes the elapsed time and all earlier null proposals.
If the event is a null proposal along
an unknown eligible \(e\), its density contributes the same rate \(r_e\)
whether \(e\) is present or absent.  Using the values of \(q_e\) and
\(r_e\) at the end of the wait gives
\[
 \Prb[\rv{Z}_e=1\mid h^-,\,\text{null proposal along }e]
 =\frac{q_e r_e}{q_e r_e+(1-q_e)r_e}=q_e .
\]
Thus observing the null proposal leaves \(q_e\) unchanged.
Both \(\exp(-\int_0^t R(s)\,ds)\) for each interval and the rate of each
recorded null proposal (\(r_e\) for an unknown eligible \(e\), and
\(1/\deg_G(u)\) for any other edge at \(u\)) are the same for all
compatible graphs.  Hence, even given the null proposals, the unknown
coins are independent and each \(\rv{Z}_e\) is
one with probability \(q_e\).
If we observe only that a proposal was made along \(e\), either a real
jump or a null proposal can account for it.  Dividing the real jump
rate by the sum of the two rates gives
\[
 \Prb[\text{real jump}\mid h^-,\,\text{proposal along }e]
 =\frac{q_e/(p\deg_G(u))}{q_e/(p\deg_G(u))+r_e}
 =\frac{q_e/(p\deg_G(u))}{1/\deg_G(u)}=\frac{q_e}{p}.
\]
The proposals therefore have the same distribution as those generated
by the rate-one clock and uniform neighbor choice in Algorithm~\ref{alg:walk}.
Each of the \(k_+(u)\) cached retained edges has jump rate \(1/(p\deg_G(u))\).
Their total jump rate is \(k_+(u)/(p\deg_G(u))\), the rate of the first
clock in Algorithm~\ref{alg:walk}.  Since the rates are equal, a jump
along a cached edge chooses uniformly among those \(k_+(u)\) edges.
Auxiliary random draws can also be generated conditional on the recorded
events.  For example, given a rejection along an unknown eligible edge,
the uniform draw for acceptance is uniform on the part of \([0,1]\) above
\(q_e/p\), for every compatible graph.

For \(\textsc{Expose}(u)\), condition on the history \(h\) just before
the call.  For each unknown eligible edge \(e\) at \(u\), let \(I_e\)
be the event that its index is selected.  Returning \(e\) requires
selecting its index and then accepting it.  Hence
\[
 \Prb[e\text{ is returned}\mid h]
 =\Prb[I_e\mid h]\,\Prb[e\text{ is accepted}\mid I_e,h]
 =p\cdot\frac{q_e}{p}=q_e.
\]
The decisions to return edges are independent across unknown edges,
just as their coins are in the
pre-sampled graph given the history.  Revealing them leaves the other
unknown coins independent.  Thus the next reply has the same distribution
in both experiments, and the remaining coins have the same conditional
distribution after the reply.  Induction over the calls gives the same
joint distribution for the entire sequence of replies.

\paragraph{Query costs for a fixed \(H\in\mathcal G\).}
Now condition on \(\rv{H}=H\) as well as on the preceding history.
In the coupling the null-proposal rates at a vertex sum to at most one.
The real jump rate at \(u\) is \(k_H(u)/(p\deg_G(u))\le2\).  Since the
stopping clock has mean \(L/4\), a walk generates \(O(L)\) events in
expectation.  Including the initial degree query and \(O(1)\) queries per
event, the expected cost is \(O(L)\) since \(L\ge1\).

For the first exposure of \(u\), condition initially on the preceding
history \(h\) alone.  If an unknown eligible edge \(e\) at \(u\) is
retained, its index must be selected, so \(\{\rv{Z}_e=1\}\subseteq I_e\).
Partitioning \(I_e\) according to the value of the coin gives
\[
 \Prb[I_e,\,\rv{Z}_e=0\mid h]
 =\Prb[I_e\mid h]-\Prb[\rv{Z}_e=1\mid h]=p-q_e.
\]
If the fixed graph \(H\) omits \(e\), we need the selection probability
conditional on its coin being zero.  It is
\[
 \Prb[I_e\mid h,\,\rv{Z}_e=0]
 =\frac{\Prb[I_e,\,\rv{Z}_e=0\mid h]}{\Prb[\rv{Z}_e=0\mid h]}
 =\frac{p-q_e}{1-q_e}\le p .
\]
Given \(h\), the pairs \((\1_{I_e},\rv{Z}_e)\) are independent across
unknown eligible edges at \(u\), and the draws for this call are
independent of the unknown coins away from \(u\).  Thus fixing all the
coins in \(H\) leaves the selection probability \((p-q_e)/(1-q_e)\)
for an unknown edge absent from \(H\).  Indices of known or ineligible
edges are selected with probability \(p\).
Every edge at \(u\) that was unknown before the call and is present in
\(H\) has its index selected with probability one.  There are at most
\(k_H(u)\) such edges, and each remaining index is selected with
probability at most \(p\).  The conditional expected number of queried
indices is therefore at most \(k_H(u)+p\deg_G(u)\).  Each index costs \(O(1)\) queries.
Including the degree query at \(u\) gives \(O(1+p\deg_G(u))\) queries
because \(H\in\mathcal G\).  Subsequent calls use the cached neighborhood.
\end{proof}

\subsection{Testing ambiguity by sampling conditional labels}
\label{sparse:sec:ambiguity-estimation}

If \(H\in\mathcal G\) and \(\tau_w(\mathcal H_H)\ge\rho/16\),
Lemma~\ref{lem:ug-multiscale-ambiguity} applied at \(\rho/16\) gives a
scale \(r\) at which the seeds \(s\) with \(\mu_{L_r}^H(s)\ge1/4\)
have total \(\pi\)-mass at least \(r/4\).  The test at a single seed
applies to every \(H\in\mathcal G\).  Fix such an \(H\), a seed \(s\in V\),
and a mean walk length \(L\ge1\).  To test \(\mu_L^H(s)\), we sample an
endpoint and compare two independent predictions at that endpoint for
each initial label \(a\).  Averaged over endpoints, the disagreement
probability for \(a\) lies between \(\operatorname{Amb}(p^a_{L,s})\)
and \(2\operatorname{Amb}(p^a_{L,s})\), where \(p^a_{L,s}\) is the joint
vertex-label distribution for the geometrically stopped \(P_H\)-walk.
To generate the second prediction even at a rarely reached endpoint,
we use the forward-push/reverse-walk decomposition of bidirectional
PageRank~\cite{LofgrenBanerjeeGoel}.

Think of starting one unit of probability mass at \(s\) for each initial
label \(a\in[Q]\).  At each vertex \(v\in V\), the forward exploration
keeps two nonnegative \(Q\times Q\) matrices, the reserve \(A_v\) and
the residual \(R_v\).  Row \(a\) tracks the walks with initial label
\(a\).  For \(b\in[Q]\), the entry \(A_v(a,b)\) counts the probability
mass already assigned to stopping at \(v\) with label \(b\).  The entry
\(R_v(a,b)\) is the mass at \((v,b)\) whose stopping location is still
to be determined.
A push at \(v\) splits the mass in every entry of \(R_v\) according to the
next decision of the geometric walk.  A fraction \(1/(L+1)\) stops at \(v\), and the rest
takes one step, carrying its label along the transition.  Since the
choice of transition is independent of the label, every row of \(A_v\) has
the same sum \(a_v\in[0,1]\), and every row of \(R_v\) has the same sum
\(r_v\in[0,1]\).

For a threshold \(h>0\), the forward exploration stops once \(r_y\le h\pi(y)\)
at every vertex \(y\).  To see how this helps us sample a label at a given
endpoint \(v\), consider the residual mass \(r_y\) at \(y\).  Let \(u_{L,y}(v)\)
be the probability that a geometrically stopped \(P_H\)-walk from \(y\)
ends at \(v\).  By reversibility, the residual mass that will stop at \(v\)
from \(y\) is
\[
 r_y u_{L,y}(v)=\pi(v)\,u_{L,v}(y)\,\frac{r_y}{\pi(y)}.
\]
To sample a label conditional on ending at \(v\), choose between the
reserve and residual branches with probabilities proportional to \(a_v\)
and \(h\pi(v)\), respectively.  In the residual branch, start a walk at
\(v\).  It ends at \(y\) with probability \(u_{L,v}(y)\).  Accept there
with probability \(r_y/(h\pi(y))\), which is at most one by the push
threshold.  Before normalizing the branch weights, the contribution from
walks accepted at \(y\) is
\[
 \underbrace{h\pi(v)}_{\text{residual branch weight}}\,
 \underbrace{u_{L,v}(y)}_{\text{probability of ending at }y}\,
 \underbrace{\frac{r_y}{h\pi(y)}}_{\text{acceptance probability at }y}
 = r_y u_{L,y}(v).
\]
This is exactly the residual mass sent from \(y\) to \(v\) in the
forward direction.  The reserve branch accounts for the mass \(a_v\)
already assigned to \(v\).
A larger \(h\) lowers the \(pm/h\) term in the forward query bound and
raises the \(Lh\) term in the per-trial bound.  Taking \(h=\sqrt{pm/L}\)
makes the two terms equal.

\begin{procedure}{SparseAmbiguityTest\((s,L,\delta)\)}{alg:ambiguity-test}
\procedureio{A seed \(s\in V\), a mean discrete walk length \(L\ge1\), an error
parameter \(0<\delta<1/2\), and the shared oracle of Algorithms~\ref{alg:walk}
and~\ref{alg:expose}.}
  {HIGH if \(\mu_L^H(s)\ge1/4\) and LOW if \(\mu_L^H(s)\le1/40\),
with error probability at most \(\delta\) for \(H\in\mathcal G\).}
\begin{enumerate}[leftmargin=*,label=\arabic*.]
\item \textbf{Initialize.}  If \(p<1\), query the degree of \(s\) and
return LOW if \(s\) is ineligible, since a walk from \(s\) then stays put.
Set the push threshold \(h:=\sqrt{pm/L}\) and the number of trials
\(M:=\lceil100\log(2Q/\delta)\rceil\).
\item \textbf{Forward push.}  Initialize the residual \(R_s\) to the
\(Q\times Q\) identity matrix \(I_Q\) and all other residuals and all reserves
\(A_v\) to zero.  While some \(v\) has residual row sum \(r_v>h\pi(v)\),
push at the least such vertex in the input ordering.  On its first push,
call \(\textsc{Expose}(v)\) to obtain the transition probabilities \(P_H(v,w)\)
from~\eqref{sparse:eq:sparse-kernel}.  Query the degrees of newly reached
vertices as needed to compute \(\pi\).

Save the current residual matrix as \(R_v^{\rm old}\) and clear \(R_v\).
Add \(R_v^{\rm old}/(L+1)\) to \(A_v\).  For each transition \(v\to w\),
add \(\frac{L}{L+1}P_H(v,w)R_v^{\rm old}\) to \(R_w\), moving column
\(b\) to column \(\pi_{vw}(b)\) when \(w\ne v\) and leaving the columns
in place when \(w=v\).  Keep the final \(A_v\) and \(R_v\) at every
vertex for all trials.

\item \textbf{Compare labels.}  In each of the \(M\) trials, call
\(\textsc{Walk}(s,L)\) with fresh randomness to obtain
\((\rv{v},\rv{\Gamma}_0)\in V\times S_Q\).  Write \(v\) for the observed
endpoint.  For each initial label \(a\in[Q]\), the first predicted label
is \(\rv{\Gamma}_0(a)\).  Draw a second label \(\rv{c}_a\in[Q]\) at \(v\)
as follows.  Choose between the two branches with probabilities
proportional to \(a_v\) and \(h\pi(v)\), respectively.  Repeat with fresh
randomness until an attempt accepts.
\begin{enumerate}[label=(\alph*),nosep]
 \item \emph{Reserve.} Accept and draw each \(\rv{c}_a\)
 from the normalized reserve row \(A_v(a,\cdot)/a_v\).
 \item \emph{Residual.} Call \(\textsc{Walk}(v,L)\) to draw a reverse endpoint
 \(\rv{y}\in V\) and a transport permutation \(\rv{\Gamma}\in S_Q\).
 Accept with probability \(r_{\rv{y}}/(h\pi(\rv{y}))\).  On acceptance,
 draw a label \(\rv{d}_a\in[Q]\) for each \(a\) from the normalized residual
 row \(R_{\rv{y}}(a,\cdot)/r_{\rv{y}}\).  Set \(\rv{c}_a:=\rv{\Gamma}^{-1}(\rv{d}_a)\),
 the label at \(v\) that this path transports to \(\rv{d}_a\) at \(\rv{y}\).
\end{enumerate}
Use the same acceptance decision for every initial label.  For each \(a\),
record whether \(\rv{\Gamma}_0(a)\ne\rv{c}_a\).

\item \textbf{Decide.}  Let \(\widehat{\rv{J}}_a\in[0,1]\) be the fraction
of trials with a disagreement for initial label \(a\).  Return HIGH if
\(\min_{a\in[Q]}\widehat{\rv{J}}_a\ge1/8\), and LOW otherwise.
\end{enumerate}
\end{procedure}

\begin{proposition}[Ambiguity testing from any seed]
\label{sparse:prop:shared-ambiguity}
Fix \(H\in\mathcal G\), a seed \(s\in V\), \(L\ge1\), and \(0<\delta<1/2\).
Conditional on \(\rv{H}=H\) and any preceding oracle history,
\(\textup{\textsc{SparseAmbiguityTest}}(s,L,\delta)\) reports
\textup{\textsc{HIGH}} with probability at least \(1-\delta\) if
\(\mu_L^H(s)\ge1/4\), and \textup{\textsc{LOW}} with probability
at least \(1-\delta\) if \(\mu_L^H(s)\le1/40\).
Its expected number of original-graph queries under the same conditioning is
\begin{equation}
 O\left((\sqrt{Lpm}+L)\log\frac{Q}{\delta}\right).
 \label{sparse:eq:shared-ambiguity-cost}
\end{equation}
\end{proposition}

\begin{proof}
Fix \(H\in\mathcal G\) and condition on \(\rv{H}=H\) and the preceding
oracle history.  At an ineligible seed all walks stay at \(s\) with their
initial label, so \(\mu_L^H(s)=0\) and LOW is correct.  Suppose henceforth
that \(s\) is eligible.  Every vertex reachable from \(s\) through retained
edges is then eligible.  Let \(G_L\) be the endpoint kernel for
geometrically stopped walks on the label lift of \(P_H\), whose state
space is \(V\times[Q]\).
In this proof, \(p^a_{L,s}\) and \(u_{L,s}\) denote the endpoint
probabilities defined in Subsection~\ref{subsec:ug-setup}, now for \(P_H\).
Thus \(G_L((s,a),(v,b))=p^a_{L,s}(v,b)\).  Inverse
constraints make the lift reversible with measure \(\pi(v)/Q\) at
\((v,b)\).

\paragraph{Forward-push invariant.}
For each initial label \(a\) and endpoint-label pair \((v,j)\), the
reserve gives the probability already assigned to \((v,j)\).  The residuals
account for the remaining probability through the identity
\begin{equation}
 G_L((s,a),(v,j))=A_v(a,j)+\sum_{y,b}R_y(a,b)G_L((y,b),(v,j)).
 \label{eq:ppr-invariant}
\end{equation}
Initially \eqref{eq:ppr-invariant} holds because \(R_s=I_Q\).
Conditioning on whether the geometric walk stops or takes its first step
shows that a push preserves the identity.  The updates preserve equality
of the row sums because permuting columns leaves each row sum unchanged.
For every initial label, \(\sum_v(r_v+a_v)=1\).
Each push increases \(\sum_v a_v\) by more than
\(\frac{h}{L+1}\min_{v\in V}\pi(v)>0\).  Since \(\sum_v a_v\le1\), the
forward-push loop terminates with \(r_v\le h\pi(v)\) at every vertex.
The loop always pushes the least \(v\) with \(r_v>h\pi(v)\) in the input
ordering, so the final \(A_v,R_v\) and the set of pushed vertices
depend only on \(H,s,L,h\).

\paragraph{Charging only the first exposure.}
Let \(U_{\rm push}\subseteq V\) be the set of vertices pushed at least
once.  Just before the last push at \(v\), write \(c:=r_v\).
Immediately afterwards, \(r_v\) is
\[
 \frac{L}{L+1}P_H(v,v)c\ge c/4>h\pi(v)/4,
\]
since \(L\ge1\) and \(P_H(v,v)\ge1/2\).  Subsequent pushes at other
vertices only increase \(r_v\).  Thus
\begin{equation}
 \sum_{v\in U_{\rm push}}\pi(v)\le4/h.
 \label{sparse:eq:pushed-mass}
\end{equation}
For eligible positive-degree vertices, \(1+p\deg_G(v)=O(p\deg_G(v))\): when \(p=1\),
\(\deg_G(v)\ge1\), and when \(p<1\), \(p\deg_G(v)\ge60\ell_Q\).
Lemma~\ref{sparse:lem:deferred-oracle} therefore bounds the total expected
input queries for forward push, conditional on \(\rv{H}=H\), the seed \(s\), and the prior history, by
\begin{equation}
 O\left(1+\sum_{v\in U_{\rm push}}p\deg_G(v)\right)
 =O(1+pm/h).
 \label{sparse:eq:shared-push-queries}
\end{equation}
This includes any missing degree queries at returned retained neighbors,
of which there are at most \(k_H(v)\le2p\deg_G(v)\) per pushed vertex.
Repeated pushes use the cached neighborhood and degrees without new input
queries.  The distinction
between push computation and query complexity is already present in the
Andersen--Chung--Lang local-push method; see
Cui--Wei--Yang~\cite[Section~1]{CuiWeiYang2026}.
Repeated pushes still incur local computation.

\paragraph{Distribution of the sampled labels.}
For an endpoint \(v\in V\), a fresh call \(\textsc{Walk}(v,L)\)
returns a random endpoint \(\rv{y}\in V\) and a random permutation \(\rv{\Gamma}\in S_Q\)
transporting labels
from \(v\) to \(\rv{y}\).  Reversibility and \eqref{eq:ppr-invariant} give
\begin{equation}
 \begin{aligned}
 p^a_{L,s}(v,j)&=A_v(a,j)
              +\pi(v)\E\frac{R_{\rv{y}}(a,\rv{\Gamma}(j))}{\pi(\rv{y})},\\
 u_{L,s}(v)&=a_v+\pi(v)\E\frac{r_{\rv{y}}}{\pi(\rv{y})}.
 \end{aligned}
 \label{sparse:eq:conditional-label-decomposition}
\end{equation}
The second identity follows by summing over labels, since \(\rv{\Gamma}\) is
a bijection.  At any sampled endpoint \(v\), we have \(u_{L,s}(v)>0\).
Write \(z_v:=a_v+h\pi(v)>0\).  One attempt chooses the reserve branch
with probability \(a_v/z_v\) and the residual branch with probability \(h\pi(v)/z_v\).
Since \(r_y\le h\pi(y)\) for every \(y\), the second identity gives
\(u_{L,s}(v)\le z_v\).  For each \(a,j\), the two branches
of one attempt together accept and return
\(\rv{c}_a=j\) with probability
\[
 \frac{A_v(a,j)}{z_v}
 +\frac{\pi(v)}{z_v}\E\frac{R_{\rv{y}}(a,\rv{\Gamma}(j))}{\pi(\rv{y})}
 =\frac{p^a_{L,s}(v,j)}{z_v}.
\]
Thus it accepts with probability \(u_{L,s}(v)/z_v>0\), so repetition terminates
almost surely.  Given endpoint \(v\), the sampler returns \(\rv{c}_a=j\)
with probability \(p^a_{L,s}(v,j)/u_{L,s}(v)\).

Let \(\rv{N}_{\rm rev}(v)\in\mathbb Z_{\ge0}\) count the reverse walks
used in one trial to draw the labels \((\rv{c}_a)_{a\in[Q]}\) at \(v\).
Each attempt uses one with probability \(h\pi(v)/z_v\), so
\[
 \E \rv{N}_{\rm rev}(v)=\frac{h\pi(v)}{u_{L,s}(v)},\qquad
 \E_{\rv{v}\sim u_{L,s}}\rv{N}_{\rm rev}(\rv{v})
     =h\sum_{v:u_{L,s}(v)>0}\pi(v)\le h.
\]
This averages over endpoints for the same seed \(s\).  It allows arbitrarily
small positive endpoint probabilities and reducible kernels.
By Lemma~\ref{sparse:lem:deferred-oracle}, each walk costs \(O(L)\) queries
in expectation conditional on the history before its call.  Summing over
the reverse walks in one trial gives
\(O(L)\,\E_{\rv{v}\sim u_{L,s}}\rv{N}_{\rm rev}(\rv{v})=O(Lh)\)
expected queries in total to draw all \(Q\) labels \(\rv{c}_a\).
This uses the expected cost before each call, even when the cost of
that walk is correlated with its acceptance.

\paragraph{Error probability and total query cost.}
Conditional on the endpoint of the first walk being \(v\), the two labels
\(\rv{\Gamma}_0(a)\) and \(\rv{c}_a\) are independent draws from the
distribution with probabilities
\(p^a_{L,s}(v,j)/u_{L,s}(v)\) for \(j\in[Q]\).  The disagreement probability
\(J_a\in[0,1]\) is therefore given by
\[
 J_a:=\Prb[\rv{\Gamma}_0(a)\ne\rv{c}_a]
   =\sum_{v:u_{L,s}(v)>0}u_{L,s}(v)\left(1-\sum_j \left(\frac{p^a_{L,s}(v,j)}{u_{L,s}(v)}\right)^2\right).
\]
For any probability vector \(q=(q_j)_{j\in[Q]}\in[0,1]^Q\),
\[
 1-\max_jq_j\le1-\sum_jq_j^2\le2(1-\max_jq_j).
\]
Consequently,
\begin{equation}
 \operatorname{Amb}(p^a_{L,s})\le J_a\le2\operatorname{Amb}(p^a_{L,s}).
 \label{sparse:eq:disagreement-ambiguity}
\end{equation}
For a nearly satisfying labeling \(x:V\to[Q]\), this gives
\(J_{x(s)}\le2\operatorname{Amb}(p^{x(s)}_{L,s})\), which is at most twice
the probability that a geometrically stopped \(P_H\)-walk from \(s\) traverses an
edge violated by \(x\).

If \(\mu_L^H(s)\le1/40\), then
\(J_a\le1/20\) for at least one \(a\); if
\(\mu_L^H(s)\ge1/4\), then \(J_a\ge1/4\) for every \(a\).
Given \(\rv{H}=H\) and the seed \(s\), the matrices \(A_v,R_v\) at the
end of the forward-push loop are deterministic and the oracle
supplies fresh walk randomness, so trials are independent despite the
shared cache.  Hoeffding's inequality gives
\(\Prb[\widehat{\rv{J}}_a\ge1/8]\le e^{-9M/800}\le\delta/Q\) when
\(J_a\le1/20\), and the same bound on
\(\Prb[\widehat{\rv{J}}_a<1/8]\) when \(J_a\ge1/4\).
If \(\mu_L^H(s)\ge1/4\), take a union bound over labels;
if \(\mu_L^H(s)\le1/40\), one label suffices.
Independence between labels within a trial is not needed.

Including the initial call to \(\textsc{Walk}(s,L)\), a trial costs
\(O(L(1+h))\) expected queries.  Adding the forward-push queries
from~\eqref{sparse:eq:shared-push-queries} to the queries for the \(M\)
trials and substituting \(h=\sqrt{pm/L}\) gives
\[
 O\left(1+\frac{pm}{h}+ML(1+h)\right)
 =O\left((\sqrt{Lpm}+L)\log\frac{Q}{\delta}\right).
\]
This proves \eqref{sparse:eq:shared-ambiguity-cost} conditional on \(\rv{H}=H\)
for every \(H\in\mathcal G\), seed, and prior history.  Matrix and permutation
operations retain polynomial dependence on \(Q\) in local computation
and storage; only the graph-query bound has logarithmic dependence.
\end{proof}

\paragraph{Sampling seeds.}
For a constant \(\xi\in(0,1/10)\), use Lemma~\ref{lem:almost-uniform-edge-sampler}
on the original input graph to draw an edge \(\rv{e}\) from the
distribution \(\widetilde\mu_E\).  Choose one of its endpoints uniformly.
The resulting seed distribution
\(\widetilde\pi:V\to[0,1]\) is
\[
 \widetilde\pi(v):=\frac{1}{2}\sum_{e\in E:e\ni v}\widetilde\mu_E(e)
                 \in(1\pm\xi)\frac{\deg_G(v)}{2m}=(1\pm\xi)\pi(v).
\]
Sampling a seed has expected cost \(\tO(n/\sqrt m)\), using the original \(n\),
including isolated vertices.  No edge-sampling oracle is invoked on \(\rv{H}\).
For every \(H\in\mathcal G\), \(P_H\) has stationary distribution
\(\pi(v)=\deg_G(v)/(2m)\), so the same seed distribution
\(\widetilde\pi\) applies to every such \(H\).

\subsection{The tester and its analysis}
\label{subsec:ug-tester}

We now combine the ambiguity bounds with the local implementation to
prove Theorem~\ref{thm:ug-tolerant}.  We state the full tester, including
its query budget, and then analyze it.  Concentration of the HIGH reports
separates the two promises, and the budget controls the query cost on
every realization of the sparsifier.

\subsubsection{The tester}
\label{subsubsec:ug-algorithm}

Algorithm~\ref{alg:tolerant-tester} tests seeds at each scale on a single
sparsifier \(\rv{H}\).  For \(m>0\), let \(p\) be the sampling rate
in~\eqref{sparse:eq:sparsification-parameters}.  The sparsifier preserves
\(\pi(v)=\deg_G(v)/(2m)\) on \(\mathcal G\), so the tester can draw seeds
using the original graph at every scale.  Set \(\rho_0:=\rho/16\).
If \(\tau(\mathcal U)\ge\rho\), then with probability at least \(98/100\)
the sampled graph \(H\) satisfies \(H\in\mathcal G\) and \(\tau_w(\mathcal H_H)\ge\rho_0\),
by Lemma~\ref{sparse:lem:sparsification-promise}.  We therefore use the
scales from Lemma~\ref{lem:ug-multiscale-ambiguity} at distance \(\rho_0\),
\[
 \lambda_0:=c_\lambda\frac{\rho_0}{1+\log(1/\rho_0)},\qquad
 \mathcal Q_{\rho_0}:=\{2^j\lambda_0:j\in\mathbb Z_{\ge0},\ 2^j\lambda_0\le1\}
                 \cup\{1\}.
\]
For each \(r\in\mathcal Q_{\rho_0}\), let \(L_r\ge1\) be the expected
number of steps in a geometrically stopped walk,
\(N_r\in\mathbb Z_{\ge1}\) the number of seeds, \(\delta_r\in(0,1)\) the
error bound for each ambiguity test, and \(\theta_r\in(0,1)\) the
fraction of HIGH reports needed for
rejection.  If \(H\in\mathcal G\) and \(\tau_w(\mathcal H_H)\ge\rho_0\),
there is a scale \(r\) at which the seeds \(s\) with \(\mu_{L_r}^H(s)\ge1/4\)
have \(\pi\)-mass at least \(r/4\).  We therefore sample on the order of
\(1/r\) seeds, up to a logarithmic factor, and set the rejection threshold
and per-test error to smaller fractions of \(r\).  Specifically, let
\[
 \begin{gathered}
 L_r:=C r\lambda_0^{-2}\log(2n),\qquad
 \delta_r:=r/400,\qquad \theta_r:=r/32,\\
 N_r:=\left\lceil Cr^{-1}\log(100|\mathcal Q_{\rho_0}|)\right\rceil.
 \end{gathered}
\]
Here \(c_\lambda\) is the constant in
Lemma~\ref{lem:ug-multiscale-trace-residual}.  We choose one sufficiently
large absolute constant \(C\) for \(L_r\), the number of seeds at each scale, and the query
budget below, satisfying Lemma~\ref{lem:ug-multiscale-ambiguity} and the
concentration and cost bounds in the analysis.

Fix an absolute constant
\(\xi\in(0,1/10)\), and let
\(q_E(n,m)>0\) be a fixed numerical upper bound of order \(\tO(n/\sqrt m)\), including
constants and logarithmic factors, on the expected cost of one edge sample
from Lemma~\ref{lem:almost-uniform-edge-sampler} at accuracy \(\xi\).
For each \(H\in\mathcal G\), we will bound the conditional expected
number of queries before truncation by
\begin{equation}
 B:=C\sum_{r\in\mathcal Q_{\rho_0}}N_r
       \left[(\sqrt{L_rpm}+L_r)\log\frac{2Q}{\delta_r}+q_E(n,m)\right],
 \label{eq:ug-query-budget}
\end{equation}
using the same \(C\) as above.  The tester stops at \(\lceil100B\rceil\) queries.

\begin{procedure}{Tolerant tester}{alg:tolerant-tester}
\procedureio{Adjacency-list access to \(\mathcal U\) over \([Q]\), its sizes \(n,m\),
and tolerances \(0<\varepsilon<\rho<1\).}
  {Accept or reject, with the guarantees of Theorem~\ref{thm:ug-tolerant}.}
\begin{enumerate}[leftmargin=*,label=\arabic*.]
\item \textbf{Initialize.} If \(m=0\), accept.  Otherwise initialize one
shared local oracle with sampling rate \(p\).
\item \textbf{Enforce the budget throughout.} Abort and accept if an
exposed retained degree \(k_{\rv{H}}(v)\) exceeds \(2p\deg_G(v)\), or if the
total number of original-graph queries reaches \(\lceil100B\rceil\), with
\(B\) from~\eqref{eq:ug-query-budget}.
\item \textbf{Test the seeds.} For each \(r\in\mathcal Q_{\rho_0}\),
draw \(N_r\) fresh seeds from the distribution \(\widetilde\pi\).  To draw each seed
\(\rv{s}\in V\), sample an almost-uniform original edge and choose one
endpoint uniformly.  Call
\(\textsc{SparseAmbiguityTest}(\rv{s},L_r,\delta_r)\)
(Algorithm~\ref{alg:ambiguity-test}) on each seed, using
fresh randomness with the same shared oracle and cache.  If at least
\(\theta_rN_r\) reports at this scale are HIGH, reject.

\item \textbf{Accept} if no scale causes rejection.
\end{enumerate}
\end{procedure}

\subsubsection{Correctness and query complexity}

\begin{proof}[Proof of Theorem~\ref{thm:ug-tolerant}]
The case \(m=0\) is immediate, so assume \(m>0\).  The tester's parameters
satisfy \(\rho_0=\rho/16<1/10\).
On every \(H\in\mathcal G\) with
\(\tau_w(\mathcal H_H)\ge\rho_0\),
Lemma~\ref{lem:ug-multiscale-ambiguity} supplies a
scale \(r\in\mathcal Q_{\rho_0}\) with
\(\Prb_{\rv{s}\sim\pi}[\mu_{L_r}^H(\rv{s})\ge1/4]\ge r/4\).

\paragraph{Conditional correctness.}
First consider the reports before imposing the cutoff of \(\lceil100B\rceil\)
queries.  At scale \(r\), the tester rejects when the
fraction of HIGH reports reaches \(\theta_r=r/32\), and each ambiguity
test has error probability \(\delta_r=r/400\).
Lemma~\ref{sparse:lem:sparsification-promise} gives
\(\Prb[\rv{H}\notin\mathcal G]\le1/100\).  It also bounds the minimum,
over labelings, of the total weight of violated edges in \(\rv{H}\).
In the yes case, this minimum exceeds \(100\varepsilon m\) with probability
at most \(1/100\).  In the no case, it falls below \(\rho m/4\) with
probability at most \(1/100\).
In the yes case, condition on
\(\rv{H}=H\) for any \(H\in\mathcal G\) with \(\tau_w(\mathcal H_H)\le25\varepsilon\).
The tail bound in Lemma~\ref{lem:ug-completeness} controls the probability that
\(\mu_{L_r}^H(\rv{s})>1/40\).  For seeds \(s\) with \(\mu_{L_r}^H(s)\le1/40\),
Proposition~\ref{sparse:prop:shared-ambiguity} bounds the probability of a HIGH
report by \(\delta_r\).  Since \(\widetilde\pi\le(1+\xi)\pi\), the overall
probability of a HIGH report is at most
\[
 1000(1+\xi)\varepsilon L_r+\delta_r\le\theta_r/4.
\]
Here the last inequality follows from the theorem's tolerance assumption:
\(\lambda_0=\Theta(\lambda_\rho)\) for
\(0<\rho<1\), and
\(\varepsilon L_r/r=C\varepsilon\lambda_0^{-2}\log(2n)\) is sufficiently
small when \(\varepsilon\log n\le c\lambda_\rho^2\), with the theorem's
constant \(c\) chosen sufficiently small after fixing \(c_\lambda\) and
\(C\).  We used \(n\ge2\), which follows from \(m>0\).
In the no case, condition on \(\rv{H}=H\) for any \(H\in\mathcal G\) with
\(\tau_w(\mathcal H_H)\ge\rho_0\).  At the scale supplied by
Lemma~\ref{lem:ug-multiscale-ambiguity}, Proposition~\ref{sparse:prop:shared-ambiguity}
gives a HIGH report with probability at least \(1-\delta_r\) on each seed
\(s\) with \(\mu_{L_r}^H(s)\ge1/4\).  Since \(\widetilde\pi\ge(1-\xi)\pi\),
the overall probability of a HIGH report is at least
\((1-\xi)(r/4)(1-\delta_r)\ge r/8\).
Conditional on \(\rv{H}=H\), fresh seeds and the local oracle's fresh walk
randomness make the reports independent, despite their shared cache.
Chernoff's inequality, with a union bound over scales in the yes case,
bounds the conditional error by \(1/20\) in either case.

\paragraph{Queries and the global budget.}
For every \(H\in\mathcal G\), the cost in
Proposition~\ref{sparse:prop:shared-ambiguity} holds for each seed and every
prior history, and Lemma~\ref{lem:almost-uniform-edge-sampler} bounds the
expected cost of generating each seed.  Thus \(B\) from~\eqref{eq:ug-query-budget} bounds the
conditional expected cost before the global cap is imposed.  At scale
\(r\), the number of seeds is \(N_r=\tO(r^{-1})\) and
\(L_r=C r\lambda_0^{-2}\log(2n)\).  Multiplying the subroutine cost by
the number of seeds and adding seed generation gives
\[
 \tO\left(\sum_{r\in\mathcal Q_{\rho_0}}\frac{1}r
                \left[(\sqrt{L_rpm}+L_r)\log Q
                       +\frac{n}{\sqrt m}\right]\right).
\]
The \(N_r\) calls to \(\textsc{SparseAmbiguityTest}\) at scale \(r\) cost
\(\tO((\sqrt{pm}/(\lambda_0\sqrt r)+\lambda_0^{-2})\log Q)\).
Here \(\log(Q/\delta_r)=O(\log Q\,(1+\log(1/\rho)))\), so the
factor depending on \(\rho\) is absorbed by \(\tO\).
There are \(O(1+\log(1/\lambda_0))\) scales, and
\[
 \sum_{r\in\mathcal Q_{\rho_0}}r^{-1/2}=O(\lambda_0^{-1/2}),\qquad
 \sum_{r\in\mathcal Q_{\rho_0}}r^{-1}=O(\lambda_0^{-1}).
\]
Since \(pm=O(n\log(nQ)/\rho)\) and
\(\log(nQ)=O(\log n\log Q)\), the total is bounded by
\begin{equation}
 \begin{aligned}
 B&=\tO\left(\frac{\sqrt{pm}\log Q}{\lambda_0^{3/2}}
                  +\frac{\log Q}{\lambda_0^2}
                  +\frac{n}{\sqrt m\,\lambda_0}\right)\\
 &=\tO\left(\frac{\sqrt{n/\rho}\,(\log Q)^{3/2}}{\lambda_\rho^{3/2}}
                  +\frac{\log Q}{\lambda_\rho^2}
                  +\frac{n}{\sqrt m\,\lambda_\rho}\right).
 \end{aligned}
 \label{eq:ug-sqrt-n-total}
\end{equation}
The factor \(\log Q\) in the first two terms comes from the number of
disagreement trials per seed.  Bounding \(\sqrt{pm}\) adds a factor
\(\sqrt{\log Q}\) in the first term.  The third term counts queries for
drawing the seeds from the original graph and is independent of \(Q\).

Since the tester stops at \(\lceil100B\rceil\) queries,
\eqref{eq:ug-sqrt-n-total} also bounds its worst-case query complexity.
For every \(H\in\mathcal G\), the expected cost before truncation is at
most \(B\), so Markov's inequality bounds the conditional probability
of reaching \(\lceil100B\rceil\) by \(B/\lceil100B\rceil\le1/100\).
No individual seed or endpoint is discarded because of its cost.
The total error is at most \(2/100+1/20+1/100<1/3\), proving the theorem.
\end{proof}

\section{Query Lower Bounds}
\label{sec:query-lower-bounds}

Goldreich--Ron proved the \(\Omega(\sqrt n)\) lower bound for graphs of
degree three~\cite[Theorem~7.1]{GoldreichRonBoundedDegree}.
Kaufman--Krivelevich--Ron extended the neighbor-query bound to all
densities and observed the obstruction from isolated
vertices~\cite[Sections~4 and~5.2]{KKR}.  Their simplification of the
configuration model deletes loops and repeated edges~\cite[Section~5.1]{KKR}.
We give the details needed to keep exactly \(m\) edges in a simple graph.

\begin{proof}[Proof of Proposition~\ref{prop:query-lower-bounds}]
Fix a sufficiently large absolute constant \(C\).
For \(16\le m\le Cn\), let \(k\) be the largest even integer at most
\(\sqrt m\).  Start with \(K_{k,k}\) in the yes-case and two copies of
\(K_k\) joined by a perfect matching in the no-case.  In both cases add
a disjoint matching of \(m-k^2\) edges and fill the remaining vertices
with isolates.  Each clique leaves at least \(k(k-2)/4\) edges uncut,
so the no-instance has distance at least \(k(k-2)/(2m)>1/9\), using
\(k\ge4\) and \(m<(k+2)^2\).  Only \(O(\sqrt m)\) vertices are nonisolated.
Under a uniform relabeling, the two transcripts agree until a query
names one of these vertices, which happens with probability
\(O(q\sqrt m/n)\) in \(q\le n/2\) queries.  This gives the required
\(\Omega(n/\sqrt m)=\Omega(\sqrt n+n/\sqrt m)\) bound in this range.

Now let \(m>Cn\) and put \(N=n/2\).  We use two copies of a simple
\(N\)-vertex graph \(H\) in the no-case and its bipartite double cover
in the yes-case.  Thus each base edge \(uv\) joins \((u,i)\) to
\((v,i)\) in the first graph and to \((v,1-i)\) in the second,
for \(i\in\{0,1\}\).  Take
\[
 M=\min\left\{\lfloor m/2\rfloor,\binom N2\right\},
 \qquad r=m-2M.
\]
With \(M\) edges in \(H\), adding the edges \((u,0)(u,1)\) for the first
\(r\) base vertices gives exactly \(m\) edges in both graphs, even at
\(m=N^2\).  Indeed \(0\le r\le N\).  The yes-graph stays bipartite, and
every cut of the no-graph leaves at least \(2M\tau(H)\) edges uncut.

Choose \(H\) uniformly among simple graphs with a fixed degree sequence
of sum \(2M\) whose entries differ by at most one.  Write \(d=2M/N\).
We claim that \(\tau(H)\ge1/4\) with probability \(1-o(1)\).
In a uniform \(M\)-edge graph, the number of uncut edges of a fixed cut
is hypergeometric with mean at least \((1/2-O(1/N))M\).
The probability that some cut leaves fewer than \(M/4\) edges uncut is
therefore at most \(2^N e^{-\Omega(M)}\).  Our chosen degree sequence
maximizes the number of labeled realizations~\cite[Corollary~8]{BergerDegreeLists},
so conditioning on it costs a factor of at most
\(\binom{2M+N-1}{N-1}\le [e(d+1)]^N\), the bound on the number of
degree sequences.  Since \(M=Nd/2\) and \(C\) is sufficiently large,
the probability of \(\tau(H)<1/4\) given the chosen degree sequence is
\(o(1)\).  When \(\tau(H)\ge1/4\), the no-graph has distance at least
\((m-r)/(4m)>1/10\).

Independently swap the two vertex names above each base vertex with
probability \(1/2\).  Give the tester the unordered pair of vertex names
above each base vertex, the added matching, and the base degrees.
Use the same uniformly random ordering of base
neighbors in both copies above each vertex, placing any added matching
edge last.  We use the collision argument of
Goldreich--Ron~\cite[Lemma~7.4]{GoldreichRonBoundedDegree} and
Kaufman--Krivelevich--Ron~\cite[Lemmas~13--14]{KKR}.
Call a base vertex named once either of its two names has appeared in
a query or reply.  The transcripts can be coupled until a query first
reveals a base edge whose other endpoint has already been named.
Indeed, at an endpoint first named by the reply, the hidden
swap makes either name equally likely.  Subsequent traversals of that
edge in either direction or copy are then determined in both graphs.

We bound the collision probability for \(q\le c\sqrt N\) queries, even
if each reply also gives the reverse list index.  At most \(2q\) base
vertices have been named.  If \(d\ge N/16\), a fresh list entry hits
one of them with probability at most \(2q/(\lfloor d\rfloor-q)\), by
the uniform list ordering.  For \(d<N/16\), represent \(H\) by a uniform
pairing of its labeled list entries conditioned on simplicity.  Condition
on the pairings of base adjacency-list entries revealed by the queries.
Consider two specified entries \(a,b\) whose partners are still unknown.  From any
completion containing \(ab\), switch \(ab,xy\) to \(ay,xb\).
Of the \(2M\) oriented edges \(xy\), at most
\(2q+2\lceil d\rceil(\lceil d\rceil+2)\) are exposed, share an endpoint
with \(ab\), or would create a repeated edge.  This leaves \(\Omega(M)\)
choices, and the new partners of \(a,b\) determine the inverse switch.
Thus \(ab\) occurs with conditional probability \(O(1/M)\).
The named base vertices have at most \(2q\lceil d\rceil\) list entries in all,
giving collision probability \(O(q/N)\) per query in both degree ranges.

The transcript distributions have total variation distance \(O(q^2/N)\).
Conditioning the no-case on \(\tau(H)\ge1/4\) changes this by \(o(1)\).
Hence constant distinguishing advantage requires \(\Omega(\sqrt n)\)
queries for \(m>Cn\).  Finally, a tester using at most \(B\) queries in
expectation on each input can be truncated after \(\lfloor12B\rfloor\)
queries.  Its success probability drops by at most \(1/12\), so both
lower bounds also hold for expected query complexity.
\end{proof}

All yes-instances are bipartite, so the bound also applies to the
decreasing tolerance in Theorem~\ref{tol:thm:main} and, by putting a
transposition on each edge, to the two-label Unique Games special case.

\bibliographystyle{plain}

\begin{thebibliography}{10}

\bibitem{AFKK}
Noga Alon, Wenceslas Fernandez de~la Vega, Ravi Kannan, and Marek Karpinski.
\newblock Random sampling and approximation of {MAX-CSPs}.
\newblock {\em Journal of Computer and System Sciences}, 67(2):212--243, 2003.

\bibitem{AKKSTV}
Sanjeev Arora, Subhash Khot, Alexandra Kolla, David Steurer, Madhur Tulsiani, and Nisheeth~K. Vishnoi.
\newblock Unique games on expanding constraint graphs are easy.
\newblock In {\em Proceedings of the 40th Annual ACM Symposium on Theory of Computing (STOC)}, pages 21--28, 2008.

\bibitem{BergerDegreeLists}
Annabell Berger.
\newblock The connection between the number of realizations for degree lists and majorization, 2014.
\newblock arXiv:1212.5443v2.

\bibitem{BhattacharyyaYoshidaBook}
Arnab Bhattacharyya and Yuichi Yoshida.
\newblock {\em Property Testing: Problems and Techniques}.
\newblock Springer Singapore, 2022.

\bibitem{ChiplunkarKKMP}
Ashish Chiplunkar, Michael Kapralov, Sanjeev Khanna, Aida Mousavifar, and Yuval Peres.
\newblock Testing graph clusterability: Algorithms and lower bounds.
\newblock In {\em Proceedings of the 59th IEEE Annual Symposium on Foundations of Computer Science (FOCS)}, pages 497--508, 2018.
\newblock Full version: arXiv:1808.04807.

\bibitem{CuiWeiYang2026}
Guanyu Cui, Zhewei Wei, and Mingji Yang.
\newblock Accelerating the local push primitive for {PageRank} computation.
\newblock arXiv:2609.12076, 2026.

\bibitem{CzumajPengSohler}
Artur Czumaj, Pan Peng, and Christian Sohler.
\newblock Testing cluster structure of graphs.
\newblock In {\em Proceedings of the 47th Annual ACM Symposium on Theory of Computing (STOC)}, pages 723--732, 2015.

\bibitem{CzumajSohlerExpansion}
Artur Czumaj and Christian Sohler.
\newblock Testing expansion in bounded-degree graphs.
\newblock {\em Combinatorics, Probability and Computing}, 19(5--6):693--709, 2010.

\bibitem{DinitzKortsarzRaz}
Michael Dinitz, Guy Kortsarz, and Ran Raz.
\newblock {Label Cover} instances with large girth and the hardness of approximating basic {$k$}-spanner.
\newblock {\em ACM Transactions on Algorithms}, 12(2):25:1--25:16, 2015.

\bibitem{DinurGoldreich2025}
Irit Dinur and Oded Goldreich.
\newblock Expansion without connectivity: A property testing perspective.
\newblock Technical Report TR25-129, Electronic Colloquium on Computational Complexity, 2025.

\bibitem{EdenRosenbaum}
Talya Eden and Will Rosenbaum.
\newblock On sampling edges almost uniformly.
\newblock In {\em Proceedings of the 1st Symposium on Simplicity in Algorithms (SOSA)}, volume~61 of {\em Open Access Series in Informatics (OASIcs)}, pages 7:1--7:9, 2018.

\bibitem{FeiRubinfeld2026}
Yumou Fei and Ronitt Rubinfeld.
\newblock Testing bipartiteness in logarithmic rounds.
\newblock arXiv:2606.13583, 2026.

\bibitem{GhoshMRS}
Arijit Ghosh, Gopinath Mishra, Rahul Raychaudhury, and Sayantan Sen.
\newblock Tolerant bipartiteness testing in dense graphs.
\newblock In {\em Proceedings of the 49th International Colloquium on Automata, Languages, and Programming (ICALP)}, volume 229 of {\em Leibniz International Proceedings in Informatics (LIPIcs)}, pages 69:1--69:19, 2022.

\bibitem{GoemansWilliamson}
Michel~X. Goemans and David~P. Williamson.
\newblock Improved approximation algorithms for maximum cut and satisfiability problems using semidefinite programming.
\newblock {\em Journal of the ACM}, 42(6):1115--1145, 1995.

\bibitem{GoldreichBook}
Oded Goldreich.
\newblock {\em Introduction to Property Testing}.
\newblock Cambridge University Press, 2017.

\bibitem{GR}
Oded Goldreich and Dana Ron.
\newblock A sublinear bipartiteness tester for bounded degree graphs.
\newblock {\em Combinatorica}, 19(3):335--373, 1999.

\bibitem{GoldreichRonBoundedDegree}
Oded Goldreich and Dana Ron.
\newblock Property testing in bounded degree graphs.
\newblock {\em Algorithmica}, 32(2):302--343, 2002.

\bibitem{GuptaTalwar}
Anupam Gupta and Kunal Talwar.
\newblock Approximating unique games.
\newblock In {\em Proceedings of the 17th Annual ACM-SIAM Symposium on Discrete Algorithms (SODA)}, pages 99--106, 2006.

\bibitem{JhaKumar}
Agastya~Vibhuti Jha and Akash Kumar.
\newblock A sublinear time tester for {Max-Cut} on clusterable graphs.
\newblock In {\em Proceedings of the 51st International Colloquium on Automata, Languages, and Programming (ICALP)}, volume 297 of {\em Leibniz International Proceedings in Informatics (LIPIcs)}, pages 91:1--91:17, 2024.

\bibitem{Kac}
Mark Kac.
\newblock On the notion of recurrence in discrete stochastic processes.
\newblock {\em Bulletin of the American Mathematical Society}, 53(10):1002--1010, 1947.

\bibitem{KaleSeshadhri}
Satyen Kale and C.~Seshadhri.
\newblock An expansion tester for bounded degree graphs.
\newblock {\em SIAM Journal on Computing}, 40(3):709--720, 2011.

\bibitem{KapralovKochetkovaWrzosKaminska2026}
Michael Kapralov, Ekaterina Kochetkova, and Weronika Wrzos-Kaminska.
\newblock Spectral clustering in birthday paradox time.
\newblock In {\em Proceedings of the 2026 Annual ACM-SIAM Symposium on Discrete Algorithms (SODA)}, pages 3057--3071, 2026.

\bibitem{KKR}
Tali Kaufman, Michael Krivelevich, and Dana Ron.
\newblock Tight bounds for testing bipartiteness in general graphs.
\newblock {\em SIAM Journal on Computing}, 33(6):1441--1483, 2004.

\bibitem{KhotUGC}
Subhash Khot.
\newblock On the power of unique 2-prover 1-round games.
\newblock In {\em Proceedings of the 34th Annual ACM Symposium on Theory of Computing (STOC)}, pages 767--775, 2002.

\bibitem{KKMO}
Subhash Khot, Guy Kindler, Elchanan Mossel, and Ryan O'Donnell.
\newblock Optimal inapproximability results for {MAX-CUT} and other 2-variable {CSPs}?
\newblock {\em SIAM Journal on Computing}, 37(1):319--357, 2007.

\bibitem{Kolla}
Alexandra Kolla.
\newblock Spectral algorithms for unique games.
\newblock {\em Computational Complexity}, 20(2):177--206, 2011.

\bibitem{LofgrenBanerjeeGoel}
Peter Lofgren, Siddhartha Banerjee, and Ashish Goel.
\newblock Bidirectional {PageRank} estimation: From average-case to worst-case.
\newblock In {\em Proceedings of the 12th International Workshop on Algorithms and Models for the Web Graph (WAW)}, volume 9479 of {\em Lecture Notes in Computer Science}, pages 164--176, 2015.

\bibitem{MathieuSchudy}
Claire Mathieu and Warren Schudy.
\newblock Yet another algorithm for dense {Max Cut}: Go greedy.
\newblock In {\em Proceedings of the 19th Annual ACM-SIAM Symposium on Discrete Algorithms (SODA)}, pages 176--182, 2008.

\bibitem{PRR}
Michal Parnas, Dana Ron, and Ronitt Rubinfeld.
\newblock Tolerant property testing and distance approximation.
\newblock {\em Journal of Computer and System Sciences}, 72(6):1012--1042, 2006.

\bibitem{PengShen2026}
Pan Peng and Ranran Shen.
\newblock Testing for a large induced expander: Sublinear algorithms and query lower bounds.
\newblock arXiv:2611.04766, 2026.

\bibitem{PengYoshida}
Pan Peng and Yuichi Yoshida.
\newblock Sublinear-time algorithms for {Max Cut}, {Max E2Lin(q)}, and {Unique Label Cover} on expanders.
\newblock In {\em Proceedings of the 34th Annual ACM-SIAM Symposium on Discrete Algorithms (SODA)}, pages 4936--4965, 2023.

\bibitem{ShenZhuPengHuang2026}
Ranran Shen, Xiaoyi Zhu, Pan Peng, and Zengfeng Huang.
\newblock Sublinear spectral clustering oracle with little memory.
\newblock In {\em Proceedings of the 14th International Conference on Learning Representations (ICLR)}, 2026.

\bibitem{TrevisanUG}
Luca Trevisan.
\newblock Approximation algorithms for unique games.
\newblock {\em Theory of Computing}, 4(5):111--128, 2008.

\bibitem{YoshidaCSP}
Yuichi Yoshida.
\newblock Lower bounds on query complexity for testing bounded-degree {CSPs}.
\newblock In {\em Proceedings of the 26th Annual IEEE Conference on Computational Complexity (CCC)}, pages 34--44, 2011.
\newblock Full version: arXiv:1007.3292.

\end{thebibliography}

\end{document}